\newif\ifec
\ifec
    \documentclass[format=acmsmall, review=true]{acmart}
    \usepackage{acm-ec-25}
    \setcitestyle{authoryear}
\else
    \documentclass[11pt]{article}
    \usepackage[margin=1in]{geometry}
    \usepackage{amssymb}
    \usepackage[colorlinks=true, linkcolor=black, citecolor=blue, urlcolor=black, filecolor=blue, pdfborder={0 0 0}]{hyperref}
\fi
\usepackage{graphicx} 

\usepackage{booktabs}
\usepackage{rotating}
\usepackage{graphicx}
\usepackage{subfig}
\usepackage{pdfpages}
\usepackage{pdflscape}
\usepackage{textcomp}
\usepackage{longtable}
\usepackage{nicefrac}
\usepackage{adjustbox}	
\usepackage{bbm}
\usepackage{svg}

\usepackage{tikz}
\usepackage{tikz-cd}

\usepackage{natbib}
\usepackage[sectionbib]{bibunits}

\bibpunct{(}{)}{;}{a}{,}{,}

\usepackage{amsmath, amsfonts, amsthm}
\usepackage{accents}

\usepackage{mathtools}
\usepackage{cleveref}

\newtheorem{theorem}{Theorem}[section]
\newtheorem{proposition}[theorem]{Proposition}
\newtheorem{lemma}[theorem]{Lemma}
\newtheorem{fact}[theorem]{Fact}
\newtheorem{corollary}[theorem]{Corollary}

\theoremstyle{definition}

\newtheorem*{example*}{Example}

\usepackage{titlesec}
\titlespacing*{\section}{0pt}{1.5ex plus 1ex minus .2ex}{0.8ex plus .2ex}
\titlespacing*{\subsection}{0pt}{1.2ex plus 1ex minus .2ex}{0.8ex plus .2ex}

\newcommand{\abs}[1]{\left| #1 \right|}

\newcommand{\cE}{\mathcal{E}}

\newcommand{\cS}{\mathcal{S}}

\newcommand{\cV}{\mathcal{V}}

\newcommand{\defeq}{:=}

\newcommand{\Z}{\mathbb{Z}}
\newcommand{\N}{\mathbb{N}}

\newcommand{\floor}[1]{\lfloor {#1} \rfloor}
\newcommand{\paren}[1]{\left( {#1} \right)}

\newcommand{\sqparen}[1]{\left[ {#1} \right]}

\newcommand{\curly}[1]{\left\{ {#1} \right\}}

\DeclareMathOperator*{\argmin}{arg\,min}

\newcommand{\ch}[1]{\textcolor{olive}{[CH: {#1}]}}
\newcommand{\cd}[1]{\textcolor{red}{Cynthia: {#1}}}
\newcommand{\lh}[1]{\textcolor{orange}{Lunjia: {#1}}}
\newcommand{\nsi}[1]{\textcolor{blue}{NSI: {#1}}}
\newcommand{\jcp}[1]{\textcolor{green}{JCP: {#1}}}

\renewcommand{\ch}[1]{}
\renewcommand{\cd}[1]{}
\renewcommand{\lh}[1]{}
\renewcommand{\nsi}[1]{}
\renewcommand{\jcp}[1]{}

\usepackage{xcolor}
\usepackage[shortlabels]{enumitem}
\definecolor{navy}{rgb}{0, 0, 0.75}

\newcommand{\munderbar}[1]{\underaccent{\bar}{#1}}

\newcommand{\remove}[1]{{}}

\newcommand{\eur}{\mathsf{UR}}

\newcommand{\dgupper}{\bar{d}_G}
\newcommand{\dglower}{\munderbar{d}_G}

\newcommand{\ugupper}{\bar{u}_G}
\newcommand{\uglower}{\munderbar{u}_G}
\newcommand{\ubupper}{\bar{u}_B}
\newcommand{\ublower}{\munderbar{u}_B}

\renewcommand{\epsilon}{\varepsilon}

\crefname{fact}{Fact}{Facts}
\crefname{proposition}{Proposition}{Propositions}
\crefname{corollary}{Corollary}{Corollaries}
\crefname{lemma}{Lemma}{Lemmas}
\Crefname{fact}{Fact}{Facts}
\Crefname{proposition}{Proposition}{Propositions}
\Crefname{corollary}{Corollary}{Corollaries}
\Crefname{lemma}{Lemma}{Lemmas}

\title{Inducing Efficient and Equitable Professional Networks through Link Recommendations\thanks{This work was facilitated by the Hire Aspirations Institute at Harvard, supported in part by Alfred P.\ Sloan Foundation grant G-2017-9890.}}
\remove{
``Integration through Recommendations'',
``Link recommendations and cross group connectivity in a professional network'',
``Efficiency-Fairness Alignment in Link Recommendations on Professional Networking Platforms''
``Reducing Inequality in Professional Networks through Welfare-Aligned Link Recommendations''
}
\ifec
    \author{Submission 28}
\fi
\date{\today}

\ifec
\else
    
\begin{document}
        \author{Cynthia Dwork\thanks{Department of Computer Science, Harvard University.}
        \and
        Chris Hays\thanks{Institute for Data, Systems and Society, Massachusetts Institute for Technology.}
        \and
        Lunjia Hu\thanks{Department of Computer Science, Harvard University. Supported by the Simons Foundation Collaboration on the Theory of Algorithmic Fairness and the Harvard Center for Research on Computation and Society (CRCS).}
        \and
        Nicole Immorlica\thanks{Microsoft Research.}
        \and
        Juan Perdomo\thanks{Department of Computer Science, Harvard University. Supported by the Harvard Center for Research on Computation and Society (CRCS) and Alfred P.\ Sloan Foundation grant G-2020-13941.}}
        \maketitle
\fi

\begin{abstract}

    Professional networks are a key determinant of individuals’ labor market outcomes. They may also play a role in either exacerbating or ameliorating inequality of opportunity across demographic groups. In a theoretical model of professional network formation, we show that inequality can increase even without exogenous in-group preferences, confirming and complementing existing theoretical literature.  Increased inequality emerges from the differential leverage privileged and unprivileged individuals have in forming connections due to their asymmetric \textit{ex ante} prospects. This is a formalization of a source of inequality in the labor market which has not been previously explored.

    We next show how inequality-aware platforms may reduce inequality by subsidizing connections, through link recommendations that reduce costs, between privileged and unprivileged individuals.  Indeed, mixed-privilege connections turn out to be welfare improving, \textit{over all possible equilibria}, compared to not recommending links or recommending some smaller fraction of cross-group links. Taken together, these two findings reveal a stark reality: professional networking platforms that fail to foster integration in the link formation process risk reducing the platform's utility to its users and exacerbating existing labor market inequality.
\end{abstract}

\ifec
    \begin{document}
    \begin{titlepage}
        \maketitle
    
        \tableofcontents
    \end{titlepage}
\fi

\section{Introduction}

    \ch{Test.}

    Individuals' professional networks play a crucial role in their labor market outcomes.
    They facilitate job placement and lead to higher wages, increased productivity and greater performance \citep{jamison2022mentorship,burchardi2013economic,lin1986access,jackson2012social,ambrus2014consumption,lin2017building,bayer_place_2008,cingano_people_2012,arbex_network_2019,arrow_limited_2004,brown_informal_2013,burks_value_2015}.
    However, they may also exacerbate employment or wage inequality, lead to differential job search costs, yield uneven gains to new technology or otherwise entrench or increase economic or social stratification \citep{dimaggio_network_2012,ioannides_job_2004,pedulla_race_2019,beaman_job_2018}.
    
    Digital professional networking platforms crucially influence network formation by providing link recommendations between individuals.
    Indeed, a large fraction of all connections are formed through recommendations \citep{rajkumar_causal_2022}.
    Thus, the choices that platforms make with respect to their recommendation systems have far-reaching implications for their millions of users.

    A rich theoretical literature has explored the sources of inequality in networks.
    \citet{bolte_role_2020} show how networks can increase inequality relative to exogenous levels due to homophily, the phenomenon where ``similarity breeds connection'' \citep{mcpherson_birds_2001}.
    \citet{calvo-armengol_effects_2004} explore how costs of staying in the labor market can produce permanent disadvantages for those who start in opportunity-poor parts of a network.
    \citet{okafor_seeing_2022} demonstrates that inequality can emerge from homophily and disproportionately larger networks among members of a majority group. 
    Each of these contexts explore how the \textit{a priori} existence of homophily   
    can lead to increased inequality. 
    
    In this work, we explore a pathway through which inequality can be exacerbated \textit{even without exogenous homophily}. Homophily, in which individuals with similar access to opportunities connect, may form endogenously between privileged groups through a theoretical model of strategic link formation in a professional network.
    %
    Namely, we show that inequality can increase as a result of the differential leverage between individuals who come to the network with a starting level of privilege --- and therefore can offer greater benefits to potential connections --- and those that do not.
    One important new implication surfaced by our work is that it may not be sufficient to eliminate existing homophily to eliminate sources of network inequality.
    In other words, inequality can be \textit{self-propagating}.
    Thus, unless platforms or regulators can, in addition to inequality resulting from homophily, also address the underlying asymmetry in the leverage individuals bring to networking, inequality may continue to increase.
    
    We also explore the role a platform may play in ameliorating this form of inequality through link recommendations.
    We explore a set of natural link recommendation policies and show policies that favor cross-group connections decrease inequality relative to exogenous levels.
    Indeed, it turns out that cross-group link recommendations are \textit{welfare increasing}, relative to policies of not recommending any links at all or recommending a smaller fraction of cross-group links. 
    The additional welfare comes from a better use of opportunities on the platform, which is arguably good for the platform itself as well. 
    Our intuition is that a primary consideration for professional networking platforms is that users find opportunities through the platform (and thus have a reason to maintain their accounts). Thus, an increase in welfare may be aligned with the incentives of the platform acting in its long term interests.
    
    At a high level, our model consists of a population of individuals, each belonging to one of two groups. 
    Each individual in a group receives a random number of opportunities.
    Opportunities may be thought of as jobs, referrals, contracts, projects, or other desirable professional outcomes.\footnote{This model may also be suitable to describing the (costly-per-person) transmission of information within a network. 
    This is not to be confused with ``fixed-cost'' transmission, like social media posts that simultaneously go out to all social contacts at once.
    Our intuition is that costly-per-person transmission is the more relevant kind for professional networking.
    This is perhaps justified by the fact that fixed-cost transmission looks similar to network-wide communication, which is less dependent on network structure. 
    }
    They may come from one of two sources: first, exogenously or, second, after exogenous opportunities are drawn, from a contact in a professional network.
    Exogenous opportunities might come from an individual's educational or socio-economic background, their family, or other pre-determined sources.\footnote{We may even extend the model to consider ``off-platform'' links as having formed exogenously through pre-determined factors like an indivual's personal history. These off-platform links may be encapsulated into a person's exogenous opportunity distribution and we may assume that all links considered in this work are ``on-platform'' links. Our results can easily accommodate this framing of exogenous versus endogenous connections. See \Cref{app:preliminaries} for further discussion.} 
    The arrival of exogenous opportunities is assumed to be out of the individuals' control and unrelated to the individual's ability to make use of an opportunity.
    If an individual receives more exogenous opportunities than they can use, they pass remaining opportunities to their contacts in a professional social network.
    Sometimes, an individual may receive more opportunities than all of their social contacts can use, in which case any extra opportunities are wasted.
    Similarly, two individuals may unwittingly pass more opportunities to a single person than that person can use, in which case the extra opportunities are wasted.
    Thus, the second source of opportunities for an individual in our model is from their social contact receiving an extra opportunity and passing it to them.

    Our focus in this work is on \textit{network formation}. 
    Prior to the point at which exogenous opportunities are realized, individuals in the model will form connections with others.
    Forming each connection bears a small cost for the individuals involved but allows for the possibility that one may pass an opportunity to the other. 
    The cost may be imagined to be a search cost (\textit{e.g.}, the difficulty of finding each other) or an attentional burden (\textit{e.g.}, the time necessary to navigate a user interface or communicate).
    %
    Each individual will choose how many connections to create, and with whom, according to their own interests: balancing the cost of the connection against the probability that they receive an opportunity as a result of the connection and would not have otherwise received opportunities.

    Individuals within each group have the same opportunity distributions, and there will be one group with a higher number of 
    exogenous opportunities (the {privileged} group) and one with a lower number (the {unprivileged} group).
    Our focus will be on quantifying \textit{efficiency}, or the number of opportunities that are used, and \textit{inequality}, the ratio of the number of opportunities used within each group.

    Finally, we analyze a platform that has the limited capacity to influence outcomes 
    by making link recommendations.
    Recommendations have the effect of reducing connection costs by reducing search or attentional effort, and we analyze the effect of different policies on efficiency and inequality.
    We introduce the model in more detail and define notation in \Cref{sec:model}.
    %
    %

    \paragraph{Our contributions.} 
    We extend the model of strategic network formation from \citet{dwork_equilibria_2024} to the setting where there are heterogeneous starting exogenous opportunity distributions across two groups.
    To their model, we also add the role of the platform in facilitating connections through link recommendations, which take the form of a subsidy on the cost of connections.
    We also generalize their model to allow individuals to receive more than two opportunities with nonzero probability.

    We then show how platforms can increase labor market inequality, {even when} individuals have no \textit{a priori} in-group preferences.
    Specifically, we show that inequality, defined as the ratio of average utilities between the two groups, increases relative to exogenous levels at equilibrium in the model where all connections are organic (\emph{i.e.}, the platform does not make any recommendations). 
    This is because they enable cheap connections of which only the already privileged can take advantage.
    Thus, a ``facially neutral'' platform that does not try to intervene on the link recommendation process increases inequality.
    
    Interestingly, despite the fact that our model does not bake in any \textit{a priori} in-group preferences, homophily may emerge endogenously.
    This is because, especially when exogenous inequality between groups is a large, individuals who are privileged may only be willing to connect to other privileged individuals; connecting with a less privileged person is not worth the (communication or search) costs associated with connecting, and so facially neutral platforms may be segregated.
    Thus, our work can serve as a microfoundation for homophily in professional networking, complementing the existing literature where homophily is assumed to be pre-existing.
    
    Towards understanding what platforms can do to address inequality, we next turn our attention to their influence over the network formation process.
    Specifically, we show that the natural way of alleviating inequality, \emph{i.e.}, encouraging cross-group connections, is also the most effective way of improving total utility, over all possible equilibria, as long as the exogenous level of inequality is sufficiently large.
    The reason for this is that fewer opportunities are wasted when the network is more integrated: With more cross-group links, it is more likely that, when an opportunity is passed, it is passed to a lower-privileged person who would not otherwise receive an opportunity. This is similar to the reason affirmative action decreases inequality in \citet{bolte_role_2020}; we show the effect persists even in the presence of strategic network formation.
    
    Thus, it can be in the best interest of a purely utilitarian platform to recommend cross-group links and consequently reduce inequality. Somewhat surprisingly, we demonstrate this fairness-utility alignment  even without exogenous reasons to favor diversity: all the opportunities in our model are equivalent, so an opportunity obtained from a cross-group link does not bring any extra benefit compared to an opportunity from an in-group link. In reality, opportunities may come in a variety of types and it is often beneficial to have a wide coverage of different types of opportunities from diverse sources \citep{heidari_informational_2023}. In those cases, the incentive of a utilitarian platform to recommend cross-group links can be even stronger. {Our result provides an explanation for the advantages of integration {even without} \textit{a priori} reasons to favor diversity.}

    \paragraph{Related work.} 
    Several previous theory models have explored quantification of or remedies for inequality in networks for hiring and opportunity.

    We adapt our model from \citet{dwork_equilibria_2024}, which, among other things, quantifies inequality that can emerge among a single, homogeneous group as a result of network formation.
    Our model extends theirs by allowing for social groups with distinct levels of exogenous privilege, by allowing the network to influence incentives by making recommendations, and generalizing exogenous opportunity distributions to allow an individual to receive more than one opportunity with nonzero probability.

    \citet{bolte_role_2020} studies how exogenous inequality and \textit{ex ante} homophily within two demographic groups can lead to increased inequality. 
    In our model, there is no ex ante homophily inducing inequality --- instead, ``neutral'' networks increase inequality because privileged individuals have more leverage to form higher value connections.
    Thus, our analysis describes a different pathway through which inequality can increase in a network.
    We expect that introducing \textit{ex ante} homophily would lead to further increases in inequality in our model, but leave exploration of this to future work.
    \citet{bolte_role_2020} also quantify how affirmative action (reducing a hiring threshold for the minority group) may increase efficiency, which is complementary to our analysis of how network integration can increase efficiency.

    \citet{okafor_seeing_2022}, using a variant of the model developed in \citet{montgomery_social_1991}, studies how inequality can emerge among two demographic groups where there is no exogenous inequality, one group is larger, and there is \textit{ex ante} homophily.
    The mechanism by which inequality emerges is that the network formation process in exogenous and leads to a disproportionate number of ties going from minority group to the majority group.
    Thus, even though minority group members prefer to pass opportunities to other minority group members, the large number of connections they have to majority group members leads opportunities to disproportionately go to majority group members.
    In our model, we study how inequality can increase \textit{without} {pre-existing} homophily but \textit{with} a network formation process influenced by the incentives of individuals.

    In \citet{calvo-armengol_effects_2004}, inequality between groups emerges as a result of differential exogenous starting employment status and up-front costs of staying in the market.
    In their model, agents who do not wish to pay the up-front costs will permanently leave the market, leading others in their neighborhood to have worse prospects and to do the same.
    The starting network in their model is given (although some agents may leave the network, no new connections form), whereas ours is endogenous, allowing us to ask about how a platform may induce new connections that increase efficiency and equity.

    More broadly, our work fits in a large body of research providing microfoundations for addressing inequality in society. 
    Among them, \citet{heidari_allocating_2021} considers a model of intergenerational opportunities showing that it can be economically efficient to allocate (\textit{e.g.}, educational) opportunities to individuals of lower socioeconomic status rather than higher-performing individuals of higher socioeconomic status.
    Other work has considered how affirmative action policies may help inform firms about the capabilities of or counteract the biases against minority workers \citep{coate_will_1993,celis_effect_2021,kleinberg_selection_2018}.
    Our paper complements this literature by providing a microfoundation for cross-group link recommendations in professional networking.

    A final related line of related work considers the learning theoretic aspects of determining who to connect in social networks.
    Formalization of the problem of determining which links to form in a social network dates back to at least \citet{liben-nowell_link_2003}.
    \citet{dwork_fairness_2024} shows how to make fair link predictions based on potentially complex and evolving characteristics. 
    Our work does not consider heterogeneous link formation probabilities between individuals in a population, and extensions to address this aspect of network formation would be a valuable direction for future work.

    \paragraph{Organization of the paper.} In \Cref{sec:model} we formally define the model and notation. In \Cref{sec:equilibrium}, we establish several features of equilibria in our model. In \Cref{sec:inequality} we explore how facially neutral platforms can increase inequality in equilibrium, relative to exogenous levels. We also analyze how cross-group link recommendations can be used to counteract these effects and reduce inequality relative to exogenous levels. In \Cref{sec:welfare}, we explore the welfare implications of cross-group link recommendation policies and show that recommending cross-group links is utility improving relative to not making any recommendations or making a smaller fraction of cross-group recommendations, for all equilibria. In \Cref{sec:conclusion}, we conclude and discuss avenues for future work.

    \section{Model} \label{sec:model}

    We  analyze a population of individuals $i \in [n]$.     Individuals belong to one of two (non-empty) groups: a {\em privileged} green group, denoted  $G$, and an {\em unprivileged} blue group, denoted $B$.     We assume that $\abs{G}, \abs{B} = \Omega(n)$, but the groups need not be the same size.
    
    Individuals each seek a single opportunity. We assume all opportunities are of equal value to all individuals, normalized to $1$.
    Opportunities arrive either exogenously from nature or endogenously through the job network. 
    Exogenous opportunities are drawn from group-dependent probability distributions. 
    %
    %
    %
    For each $i \in B$ (resp.\ $i\in G$), the probability of $i$ receiving $k$ exogenous opportunities is denoted $b_{k}$ (resp.\ $g_k$). Exogenous opportunities are drawn IID conditional on group.
    Privilege relates to the exogenous distribution of opportunities. 
    Namely, we assume that $g$ stochastically dominates $b$ (i.e., for all $m \in \Z_{\geq 0}$ it holds $\sum_{\ell \ge m}g_\ell\ge\sum_{\ell \ge m}b_\ell$ with the inequality strict for some $m$).
    We will also impose assumptions, to be formalized later, ensuring that the $b_0$ is sufficiently large and that $g_0$ is not too large.
    These assumptions allow us to ensure that \textit{inequality} is sufficiently large so that our results hold \textit{over all possible equilibria} and thus for any equilibrium selection process.
    %
    %
    %
    %
    %
    If we do not wish to specify a group membership for an individual $i$, we will refer to their probability of receiving $\ell$ exogenous opportunities as $p_{i\ell}$.
    We also assume that there is a constant $C$ so that for all $\ell \geq C$, it holds $p_{i\ell} = 0$ for all $i \in [n]$.

    %
    %
    %
    %
    %

    Opportunities also arrive  through an endogenously-formed job network, creating a game whose action space for each individual is a selection of potential neighbors. 
    The timing of the game is as follows. A platform first recommends a (potentially empty) subset of links $Q_i\subseteq\{(i,j)\,:\,j\in[n]\}$ of size $\abs{Q_i}=k$ for each $i \in [n]$, which have the effect of eliminating the cost of the connection. We assume if $(i,j)\in Q_i$, then $(j,i)\in Q_j$.   Define $Q=\cup_{i\in[N]}Q_i$.  Individuals --- knowing $Q$, $[n]$, and the exogeneous distribution of opportunities for each individual but \textit{not} their specific realization --- then choose neighbors, forming a network $E\subseteq\{(i,j)\,:\,i,j\in[n]\}$.  
Let $N_i(E)=\{j \, : \, (i,j)\in E\}$ and $d_i(E)=\abs{N_i(E)}$ (when the edge set $E$ is clear from the context; we drop it from the notation). We will call edges in the network that were not recommended (\textit{i.e.}, $E \setminus Q$) \textit{organic} or \textit{non-recommended} connections. Finally,  exogenous opportunities are realized and extra ones are distributed to neighbors in the network uniformly at random.
    %
    %
    %
    %
    Specifically, if $i$ receives $\ell> 1$ exogenously opportunities, then they select $(\ell-1)$ of their $d_i$ neighbors uniformly at random and pass each of them a single opportunity (discarding leftover opportunities if $\ell-1>d_i$).
    Let
    \begin{align*}
        \mu_i(d) = \textstyle \sum_{\ell \geq 1} \min \curly{\ell - 1, d}p_{i\ell}
    \end{align*}
    be the expected number of opportunities an individual of degree $d$ passes to their neighbors. 
    Let $\mu_B(d)$ (resp.\  $\mu_G(d)$) be this expectation for a generic (\emph{i.e.}, any) blue (resp.\ green) group member.
    %
    %
    We will overload notation and write $\mu_i = \mu_i(\infty)$  when we wish to refer to the expected number of extra opportunities a member receives (sim.\ $\mu_G$ and $\mu_B$). I.e., $\mu_i(\infty)$ is the expected number of opportunities a person with arbitrarily large (hence, $\infty$) degree will pass to their neighbors.
    %
    To avoid trivialities where there is no chance a member of a group receives any extra opportunities, we will assume throughout that $\sum_{\ell \geq 1} (\ell - 1) g_\ell$ and $\sum_{\ell \geq 1} (\ell - 1) b_\ell$ are positive.

    Given a network and realization of exogenous and endogenous opportunities, an individual $i$ receives a utility of $1$ if they receive an (exogenous or endogenous) opportunity and pays a cost of $\gamma$ for each of their unrecommended neighbors in the network. Thus their \textit{ex ante} expected utility from a network $E$ is: 
    \begin{align}
    u_i(E)=\bigg(1-p_{i0}\prod_{j\in N_i(E)}\left(1-\frac{\mu_j(d_j)}{d_j}\right)\bigg)-\gamma \cdot \abs{N_i \setminus Q_i}    
    \label{eq:utilformula}
    \end{align}
where the first term is the probability of receiving an opportunity and the second is the cost of their organic connections.  Social welfare will be defined, depending on the context, as the sum of individuals' expected utilities or the minimum utility over the population. 

We study the equilibria of the induced game.  As in \citet{dwork_equilibria_2024}, we will use defection-free pairwise Nash (DFPN) as our equilibrium concept. We will characterize pure-strategy equilibria in this work (\textit{i.e.}, we will not consider players who select edges non-deterministically). This yields the property that equilibria are uniquely specified by a set of formed edges. 
    %
    %
    DFPN considers defections in which individuals may unilaterally sever links and/or pairs of individuals may add a link if it is mutually beneficial. Notably in contrast to the notion of pairwise stability studied in \citet{jackson_strategic_1996}, a single defection may involve both severing and forming links simultaneously. In this sense, DFPN is a refinement of pairwise stability that allows individuals to consider dropping current links when contemplating forming a new one, a natural option for users in online platforms which is analogous to stability in stable marriage problems. 
    Formally, an edge set $E$ is DFPN if
    \begin{enumerate}[(a)]
        \item 
        for all $i,j$ such that $(i,j) \not\in E$ and for all $S_i \subseteq N_i(E)$, $S_j \subseteq N_j(E)$
        \begin{align}
        \begin{aligned}
            0 &\geq {u_i(E \cup \{ (i,j) \} \setminus \{ (i,\ell) \; : \; \ell \in S_i \})} - {u_i(E)}, \;\text{or} \\
            0 &\geq {u_j(E \cup \{ (i,j) \} \setminus \{ (j,\ell) \; : \; \ell \in S_j \})} - {u_j(E)}, 
        \end{aligned}\label{eq:noadd}
        \end{align}
        \item and for each $i,j$ such that $(i,j) \in E$, it holds
        \begin{align}
            \begin{aligned}
            0 &\geq {u_i(E \setminus \{ (i,j) \} )} - {u_i(E)}, \;\text{and}  \\
            0 &\geq {u_j(E \setminus \{ (i,j) \} )} - {u_j(E)}.
            \end{aligned}\label{eq:nodelete}
        \end{align}
    \end{enumerate}
    Intuitively, \cref{eq:noadd} specifies that no pair of individuals would prefer to form a connection amongst themselves, dropping a (possibly empty) set of their existing connections.
    \Cref{eq:nodelete} specifies that no individual in a connection would rather sever the connection.
    For fixed parameters, we will write the set of DFPN equilibrium edge sets as $\cE$. If we wish to contrast different equilibrium edge sets for different parameters, we will specify a particular choice of parameter (\textit{e.g.}, number of recommendations $k$) using subscripts (\textit{e.g.}, $\cE_k$).
    %
    %

Our main results characterize the inequality and welfare properties of different link recommendation policies on equilibrium outcomes.  In particular, we will study policies where the platform recommends a fixed fraction of cross-group connections to each individual within a group. Since we are not assuming $\abs{B} = \abs{G}$, it may not be possible for these fractions to be the same for members of $B$ and $G$. We will denote the fraction of cross-group connections recommended to $B$ as $\rho$ where $\rho \in \{0, 1/k, \dots, 1\}$ for the number of recommendations $k$.
This implies that members of $G$ receive a $\rho \abs{B}/\abs{G}$ fraction of cross-group recommendations.
Throughout, we will assume that $\rho$ is set so as to ensure $\rho \abs{B}/\abs{G} \leq 1$. This inequality is trivially satisfied if $\abs{B} \leq \abs{G}$.
Note that Hall's marriage theorem implies that such a set of cross-group  recommendations always exists. The remaining same-group recommendations can be constructed if the relevant populations have even size, which we assume for convenience.  See \Cref{app:preliminaries} for further explanation.

\section{Equilibrium properties}
\label{sec:equilibrium}

    The equilibria of our game satisfy some useful properties, also potentially of independent interest.  The first is a type of balance condition.  It states that almost all individuals who are connected have similar probabilities of passing each other opportunities. %
    In other words, individuals more-or-less get what they give in their relationships. 
    
   The formal lemma, stated below and proved in the appendix, extends the results of \citet{dwork_equilibria_2024} to our more complex setting with recommendations and two groups. Recall that $\mu_i(d)/d$ is the probability that an individual $i$ of degree $d$ passes one of their connections an opportunity.  Define constant $C(\gamma,k)= 2 (\gamma^{-1} + k)^2(\gamma^{-1} + k + 1)^2$.  
   %

        \begin{lemma} \label{lem:reciprocity-w-recs}
        For all equilibrium  $E \in \mathcal{E}$, all $\varepsilon > 0$, and the constant $C(\gamma,k)$, there exists a set of individuals $S$ with $\abs{S} \geq n - C(\gamma,k)/\varepsilon$ such that: 
        \begin{enumerate}[(a)]
            \item For all individuals $i\in S$ and $j\in N_i(E)$,
            \begin{align}
                \frac{\mu_i(d_i)}{d_i} \geq \frac{{\mu_j(d_j+1)}}{d_j + 1}  -  \varepsilon, \label{eq:genreciprocity}
            \end{align}
            and $j \in S$.
        Also, if $i$ and $j$ are members of the same group and $\varepsilon$ is sufficiently small, then it also holds that ${\mu_i(d_i)}/{d_i} \leq {\mu_j(d_j-1)}/{(d_j-1)} + \varepsilon$.
        \item For all $i \in S$, there exists some $j \in S$ such that $(i,j) \not \in E$, $d_i = d_j$ and $i$ and $j$ are members of the same group (implying that \cref{eq:genreciprocity} holds). 
        \end{enumerate}
        Moreover, if $\varepsilon$ is small enough (i.e., upper bounded by a constant not depending on $n$), the inequalities hold exactly without the additive $\varepsilon$.  
    \end{lemma}

    Intuitively, part (a) of \Cref{lem:reciprocity-w-recs} says that each individual $i$ participating in a connection $(i,j)$ in equilibrium satisfies an inequality where their connection $j$ has probability of passing them an opportunity no less than the probability $i$ would pass $j$ an opportunity 
    if they had degree one larger. 
    It additionally says that, for connected members of the same group, each individual exactly satisfies the condition that they have probability of passing their connection an opportunity no more than if their connection had one fewer connection, respectively.\footnote{This condition is not trivially satisfied by the fact that $\mu_j(d_j)/d_j \geq \mu_i(d_i + 1)/(d_i + 1) - \varepsilon$ since the $d_i, d_j$ are determined by {equilibrium conditions}, and may not continue to hold if individuals' degrees change.}
    Part (b) further states that each individual has an outside option: they \textit{could} be (but are not) connected to someone else of their same degree with an approximately equal expected number of extra opportunities.

    The result above is quantified by a slack term $\varepsilon$.
    It holds for all $\varepsilon$, but there is a trade-off between the tightness of \cref{eq:genreciprocity} and the size of $S$.
    Making $\varepsilon$ smaller makes the inequality tighter but decreases the size of the set $S$.
    Making $\varepsilon$ larger makes the inequality looser but increases the size of the set $S$.
    Depending on the situation, it may be desirable to choose different values of $\varepsilon$, but the choice of $\varepsilon$ is \textit{not} a parameter influencing the set of possible equilibria.
    Instead, it is a analyst-defined choice for finding a $n - O(1)$ (as long as $\varepsilon$ is not dependent on $n$) set of individuals in any equilibrium who satisfy the balance conditions in the lemma.

    We also observe that by the fact that $j \in S$, \cref{eq:genreciprocity} holds when we switch the roles of $i$ and $j$. That is, $\mu_j(d_j)/d_j \geq \mu_i(d_i+1)/(d_i + 1) - \varepsilon$.
    Together, \cref{eq:genreciprocity} and the inequality above ensure that all neighbors of $i$ satisfy a balance condition ensuring that neither individual in a connection will have substantially lower probability of passing the other an opportunity.

    It is also not hard to see that the form of the utility function as well as the equilibrium conditions imply the following two facts:
    \begin{fact}
    \label{lem:all-edges-form}
        In any equilibrium $E \in \cE$, all recommended edges form, \emph{i.e.}, $Q \subseteq E$.
    \end{fact} 

    \begin{fact}
        In any equilibrium $E\in\cE$, all individuals have bounded degree. In particular, $d_i\leq1/\gamma + k$ for all $i\in[n]$. 
    \end{fact}

    %

    We omit the proofs of these elementary facts and remark that, although we assume that recommendations eliminate connection costs in our model, we could relax this modeling choice and assume instead that recommendations only \textit{reduce} (but do not eliminate) the cost of connections.
    Our results in this paper continue to hold as long as recommended connection costs are sufficiently small (\emph{i.e.}, small enough that members of the privileged group can be incentivized to form connections with members of the unprivileged group in equilibrium). 
    We also note that central planners might extract even greater social benefits in cases where they can subsidize links by more than the connection cost, yielding payments (negative costs) for links formed.
    We leave exploration of these topics for future work.


    Finally, we establish that there always exists an equilibrium, as long as exogenous inequality is large enough. 

    \begin{fact} \label{fact:eqbm-exists}
        For all $g, b, k, \rho$ and large enough $n$ such that $\gamma > b_0(1-b_0 - b_1)$, it holds $\cE \neq \varnothing$.
    \end{fact}

    We can construct equilibria in our model by letting each individual in a group have the same degree as every other member of their group.
    Individuals in different groups will have different degrees (according to their privilege) distributions.
    Setting $\gamma > b_0(1-b_0 - b_1)$ ensures that members of $B$ will only form connections through recommendations, which considerably simplifies the analysis and conforms with the parameter regime considered in this work (where members of $B$ are sufficiently unprivileged).\footnote{We remark that our model in this paper concerns \textit{economically meaningful} connections, \textit{i.e.}, connections that carry with them some possibility that the individuals incident to the connection benefit. Our intuition is that the number of economically meaningful connections an individual has may be much lower than their number of ``connections'' on professional networking platforms might indicate. Thus, depending on the situation, the number of connections in our model may or may not exactly correspond to connections as defined by professional networking platforms. In general, the assumption that individuals must be sufficiently unprivileged so as to form few organic connections is an analytical convenience that is not necessary for the results to hold.}
    We leave for future work the question of whether there exist equilibria for any possible setting of the parameters.

    \section{Inequality analysis}
    \label{sec:inequality}

    We next explore inequality at equilibrium in our model. 
    In \Cref{sec:faciallyneutral}, we show how the existence of a platform can result in increased inequality, even without {pre-existing} homophily in the network formation process.
    Then, in \Cref{sec:cross-group-recs}, we analyze the ability for platforms to ameliorate the extra inequality they created by selectively subsidizing cross-group links.
    
    \subsection{Facially neutral platforms increase inequality.} \label{sec:faciallyneutral}

    Here, we analyze the network formation process under a platform does that not subsidize any connections via recommendations, \emph{i.e.}, $k=0$.
    We call these ``facially neutral'' because they do not intervene on behalf of any particular connections.
    As we will see, facially neutral platforms increase inequality.
    To quantify inequality, we will use a simple with-in group average. Formally, for an edge set $E$, define the \textit{utility ratio} between greens and blues as
    \begin{align*}
        \eur(E) = \frac{{\abs{G}}^{-1}\sum_{i \in G} {u_i(E)}}{{\abs{B}}^{-1} \sum_{i \in B} {u_i(E)}}
    \end{align*}
    %
    We will want to compare this quantity to the exogenous utility ratio, \emph{i.e.}, the ratio that would be obtained in the absence of the platform and hence any networking:
    \begin{align*}
        \eur(\varnothing) \defeq \frac{1 - g_0}{1-b_0}.
    \end{align*}
    %
    %
    We will sometimes refer to $\eur(\varnothing)$ as the \textit{exogenous utility ratio} and $\eur(E)$ for $E \in \cE$ as the \textit{endogenous utility ratio}.  A high utility ratio implies the privileged group fares much better than the unprivileged one, indicating high inequality.  A utility ratio of $1$ indicates equality of group welfare.

    Our simple result in this section states that for exogenous inequality sufficiently large, the existence of a facially neutral platform only makes inequality worse and that, moreover, only the privileged group benefits from the platform.  This holds for exogenous opportunity distributions where green members have a much larger chance of receiving extra exogenous opportunities, which is natural given their privilege. We also need green members to have some lower-bounded probability of receiving no exogenous opportunities. 

    \begin{proposition}
    \label{prop:ineqworse}
          Suppose $k=0$ and there exist at least two green and one blue individual(s). Then for any exogenous opportunity distributions $g$ and $b$ such that $b_0  (1-b_0-b_1)< g_0(1-g_0-g_1) $ and edge cost $\gamma \in (  b_0  (1-b_0-b_1), g_0(1-g_0-g_1))$,  
          \begin{align*}
              \eur(E) > \eur(\varnothing).
          \end{align*}
          Moreover, the benefits of the platform accrue exclusively to the privileged group. That is, if for some $i \in [n]$ it holds that ${u_i(E)} > 1 - p_{i0}$, then $i \in G$.
    \end{proposition}

    \begin{proof}
        We argue for any $(i,j)\in E\in \cE$, both $i$ and $j$ are from the privileged green group.  First note blue group individuals $i$ have probability at most $b_0(1-b_0-b_1)$ of providing an exogenous opportunity to any neighbor.  As this is less than the edge cost, no edges form to blue group individuals and hence their utility is the same in any equilibrium network as it is in the empty network. Now suppose for contradiction that no edges form between green group individuals either.  Consider two green individuals $i$ and $j$.  By forming the edge $(i,j)$, individual $i$ can obtain utility $u_i(\{(i,j)\}) = (1-g_0(1-g_0-g_1))-\gamma$ since they only fail to receive an opportunity if they themselves receive none (probability $g_0$) and $j$ receives no extra opportunity (probability $g_0+g_1$).  This is greater than their utility $1-g_0$ from the empty network so long as $\gamma<g_0(1-g_0-g_1)$. Since the situation is symmetric, $j$ also wants to form edge $(i,j)$ and hence the empty network is not an equilibrium.  Since any green can guarantee utility $1-g_0$, the average green utility in a non-empty equilibrium network must be strictly higher than that in an empty network, implying the strict inequality in the theorem statement.
    \end{proof}

    The conditions in \Cref{prop:ineqworse} are sufficient but not necessary for inequality to increase in the model: 
    Under other relations between $b$ and $g$ the platform may still induce more inequality relative to their exogenous levels.
    
    We illustrate \Cref{prop:ineqworse} in  \Cref{fig:inequality_ratio_comparison_k0}.
    In the figure, we vary the exogenous utility ratio $\eur(\varnothing) = (1 - g_0) / (1-b_0)$ on the horizontal axis and a lower bound (in blue) and upper bound (in red) on the endogenous utility ratio $\eur(E), E \in \cE$ on the vertical axis.
    Throughout this work, in each of the plots, bounds on utilities are asymptotic (as $n \to \infty$) and set $\abs{G} = \abs{B}$. 
    The non-asymptotic bounds are within additive $O(n^{-1})$ of the asymptotic bounds. 
    For simplicity, we also set $p_{i\ell }= 0$ for all $\ell \not\in \{ 0, 2 \}$ and $i \in [n]$ so that exogenous opportunity distributions collapse to a one-dimensional parameter.   %
    The black dashed line indicates where exogenous and endogenous utility ratios are equal. 
    Anything above the line indicates greater endogenous utility than exogenous utility.
    Each line in the plot corresponds to a different value of $g_0$.
    Thus, for each line, moving right along the horizontal axis corresponds to increasing the value of $g_0$.
    Discontinuities in the plot occur as a result of changes in our bounds for the set of feasible equilibria as we sweep over different parameter values.\footnote{As the value of $b_0$ increases sweeping to the right on the horizontal axis, the bounds on the set of degrees for members of both groups change.
    Since degree bounds are integer values, changes to the degree bounds lead to discontinuous changes in functions of the degree bounds, such as the endogenous utility ratio on the vertical axis.
    In particular, consider $g_0 = 0.25$, around the $\eur(\varnothing) \in [7.5,10]$ where the upper bound discontinuously jumps up and then down again.
    Before the step up, members of $B$ can only have degree $2$ in equilibrium.
    At the step up, they may have degree $2$ \textit{or} $1$.
    After the step down, they may only have degree $1$.
    Each of these changes lead to significant changes to the endogenous utility ratio, since the utilities of members of $B$ is already small for these parameter settings and changes to their utility can have large effects through the denominator of the endogenous utility ratio. Hence, we see the sharp jump up and then down again.
    The changes in degree bounds for the lower bound on the endogenous utility ratio in the same range (and for all of the other ``spikes'' in \Cref{fig:inequality_ratio_comparison_k0}) lead to the corresponding discontinuous step down and then up.}

    As we can see in \Cref{fig:inequality_ratio_comparison_k0}, for sufficiently large values of $\eur(\varnothing)$, the lower bound on inequality for each value of $g_0$ becomes larger than the break-even value where $\eur(\varnothing) = \eur(E), E \in \cE$.
    Thus, for all equilibria, the endogenous utility ratio becomes greater than the exogenous utility ratio.
    Depending on the value of $g_0$, the point at which $\eur(E) > \eur(\varnothing)$ for all $E \in \cE$ occurs at different points.
    For $g_0$ small (lighter shades), it occurs at higher values of inequality: in this plot, when $g_0 = 0.25$, this happens close to $\eur(\varnothing) = 20$.
    For $g_0$ large (darker shades of red and blue), this happens at low values of inequality: in this plot, when $g_0 = 0.75$, it occurs at $\eur(\varnothing) < 5$.
    The reason for this is that, when $g_0$ is low, $b_0$ is also relatively low, so members of $B$ may form more connections at low values of $\eur(\varnothing)$, resulting in inequality lower bounds below the break-even line.
    When $g_0$ is higher, $b_0$ is also higher, so members of $B$ form few or no connections at low values of $\eur(\varnothing)$, and thus only members of $G$ benefit from the network, leading to increased endogenous utility ratio.

    \begin{figure}
        \centering
        \includegraphics[width=\linewidth]{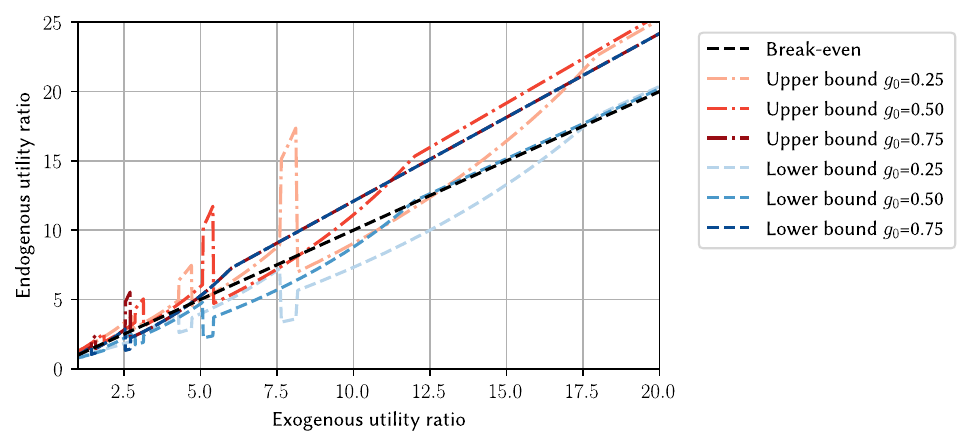}
        \caption{For large enough exogenous utility ratio, a network without link recommendations increases inequality, relative to exogenous levels. In the plot above, for each line, we fixed $g$ and vary $b$ plotting the value of $\eur(\varnothing) = (1-g_0)/(1-b_0)$ on the horizontal axis. On the vertical axis, we plot upper bounds (in red) and lower bounds (in blue) for the utility ratio $(\eur(E), E \in \cE)$ achieved for all equilibria at these parameter values. We set $\gamma = 0.04$.}
        \label{fig:inequality_ratio_comparison_k0}
    \end{figure}

    Before we conclude this subsection, it is worth comparing \Cref{prop:ineqworse} to existing results on inequality in the referrals and labor markets literature.
    An important point of departure between our work and prior work is that there are no exogenous in-group preferences between individuals.
    Exogenous in-group preferences are key to the results in \citet{bolte_role_2020} and \citet{okafor_seeing_2022}, where increased inequality emerges from the information advantage the privileged group has by virtue of receiving more referrals and the draining of opportunities from a minority group to the majority group, respectively.
    Increased inequality in our model emerges even without exogenous homophily.
    Instead, increased inequality emerges from the fact that individuals with greater privilege also have more leverage in forming connections. 
    In particular, the equilibrium conditions \cref{eq:noadd,eq:nodelete} imply that, in a sufficiently large population, almost all individuals can connect to others who are at least as valuable to them as they are to others.
    This, to our knowledge, a novel formalization for how inequality can emerge in a professional network and suggests that the existence of facially neutral platforms may still increase inequality, even if individuals' in-group biases can be eliminated.
    Of course, additionally introducing exogenous homophily in our model would only induce greater inequality at equilibrium.
    We leave deeper exploration of the interactions between exogenous homophily and the greater leverage of privileged individuals for future work.

    \subsection{Cross-group recommendations reduce inequality.} \label{sec:cross-group-recs}

    We now consider a platform that intervenes in the network formation process by selectively subsidizing connections, \emph{i.e.}, choosing $k > 0$ and making $k$ edge recommendations per person.

    In our first result in this section, \Cref{prop:cg_ineq}, we establish that for exogenous utility ratio sufficiently large, if a platform recommends all cross-group link recommendations, then inequality in any equilibrium will be strictly \textit{less} than exogenous levels.
    Together with \Cref{prop:ineqworse}, this implies that a policy recommending cross-group links yields lower inequality, for all equilibria, than any equilibrium induced by a platform that does not make link recommendations.
    We formalize this statement in \Cref{corr:cg_ineq}.
    Recall that we denote the fraction of cross-group connections recommended to $B$ as $\rho$ and the corresponding fraction for members of $G$ as $\rho \abs{B}/\abs{G}$, where we assume that $\rho$ is set so as to ensure $\rho \abs{B}/\abs{G} \leq 1$.
    In this section, we choose $\rho = 1$, so all edge recommendations are cross-group. We use $\cE_k$ to denote the set of all equilibria induced by any $k$ cross-group recommendations per person.

    \begin{proposition}
    \label{prop:cg_ineq}
        For any $k>0$, there exist constants $\munderbar{b}_0 < 1, \munderbar{n}$ such that for all $n \geq \munderbar{n}$, $g$, and $b$ satisfying $b_0 \geq \munderbar{b}_0$, there exists constant $\bar \gamma > 0$ such that for all edge cost $\gamma \in (0, \bar \gamma)$ and all equilibria $E \in \cE_k$, it holds 
        \begin{align*}
          \eur(E) < \eur(\varnothing).
        \end{align*}
    \end{proposition}

    Intuitively, the result follows from the fact that unprivileged individuals benefit more from connections to privileged individuals than privileged individuals benefit from connections to other privileged individuals.

    We note that the intuitive fact in our model that more cross-group connections leads to less inequality is not necessarily shared by other models of inequality in labor markets. For example, in \citet{okafor_seeing_2022}, networks that are totally segregated yield equal labor market outcomes, while networks that are partially integrated lead to inequality.
    
    In \Cref{fig:inequality_ratio_comparison}, we plot upper and lower bounds on inequality under different values of the exogenous utility ratio. 
    As in \Cref{fig:inequality_ratio_comparison_k0}, the red lines correspond to upper bounds and the blue lines correspond to lower bounds and different shades of a color indicate different values of $g_0$, which (due to the fact that $g_\ell = 0$ for all $\ell \neq 0, 2$) uniquely determines the exogenous opportunity distribution of members of $G$.
    As before, the discontinuities in the plot indicate points where the set of feasible equilibria change, perhaps by allowing more or fewer connections for members of $B$.
    
    We can see that for each of the values of $g_0$ in the plot, once $\eur(\varnothing)$ is more than about 3, the upper bound on utility is below the break-even line, and therefore the endogenous utility ratio is guaranteed to be less than exogenous utility ratio. 
    The point at which $\eur(E) < \eur(\varnothing)$ happens sooner for lower values of $g_0$ (displayed by lighter shades of red).
    This is because members of $G$ are likely to receive an extra opportunity and thus the benefits to members of $B$ are large.
    Indeed, for lower values of $g_0$, the endogenous utility ratio remains persistently lower than at higher values of $g_0$ as $\eur(\varnothing)$ grows large, since the members of $B$ benefit relatively more when members of $G$ are likely to receive an extra exogenous opportunity.
    Finally, we note that even a single cross-group link recommendation (we set $k=1$ in \Cref{fig:inequality_ratio_comparison}) results in the endogenous utility ratio growing sublinearly in exogenous utility. 
    This demonstrates the increasing value of cross-group link recommendations for reducing endogenous inequality at greater levels of the exogenous inequality.

    \begin{figure}
        \centering
        \includegraphics[width=\textwidth]{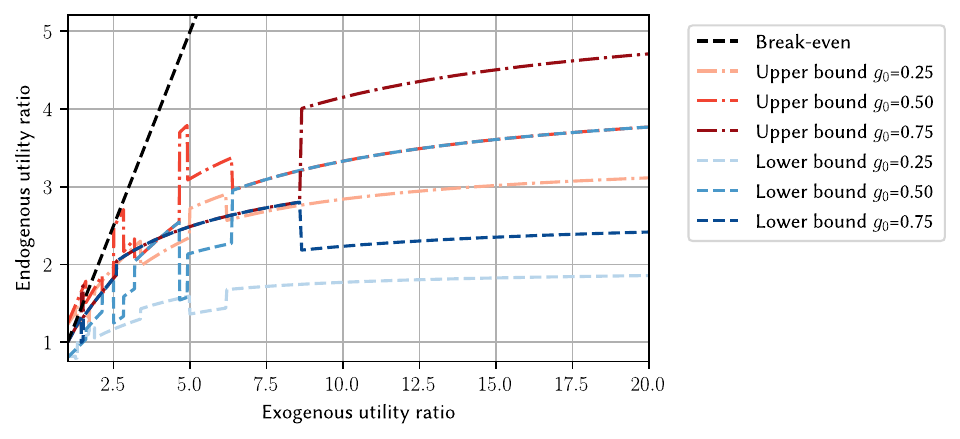}
        \caption{For large enough exogenous utility ratio, cross-group link recommendations reduce inequality, relative to exogenous levels. In the plot above, for each line, we fix $g$ and vary $b$, plotting the value of $\eur(\varnothing) = (1-g_0)/(1-b_0)$ on the horizontal axis. On the vertical axis, we plot upper bounds for the utility ratio ($\eur(E)$, $E \in \cE$) achieved for all equilibria at these parameter values. We set $\gamma = 0.04$, $k=1$ and $\rho = 1$. }
        \label{fig:inequality_ratio_comparison}
    \end{figure}
    
    A direct implication of \Cref{prop:cg_ineq} with \Cref{prop:ineqworse} is that cross-group link recommendation policies reduce inequality relative to all equilibria resulting from facially neutral platforms.
    We state this formally next.

    \begin{corollary} \label{corr:cg_ineq}
         For parameters $g, b, \gamma, n, k, \rho$ satisfying the conditions in \Cref{prop:cg_ineq}, let $\cE_k$ be the set of associated equilibria. Let $\cE_0$ be the set of equilibria with the same parameters except that the platform makes no recommendations. Then for all $E_k \in \cE_{k}$ and $E_0 \in \cE_0$, it holds 
        \begin{align*}
            \eur(E_k) < \eur(E_0).
        \end{align*}
    \end{corollary}
    The result follows from the fact that under \Cref{prop:ineqworse}, it holds $\eur(\varnothing) < \eur(E_0)$ and under \Cref{prop:cg_ineq}, $\eur(E_k) < \eur(\varnothing)$.
    In \Cref{fig:inequality_ratio_comparison_k0,fig:inequality_ratio_comparison}, the result is illustrated by the fact that the upper bounds in \Cref{fig:inequality_ratio_comparison} are below the lower bounds in \Cref{fig:inequality_ratio_comparison_k0}.

    \section{Social welfare analysis}
    \label{sec:welfare}
    
    Next, we analyze the  social welfare at equilibrium under different link recommendation policies.
    We quantify social welfare as either the average utility of all individuals in the population (a utilitarian notion) or the minimum utility of any individual in the population (a Rawlsian notion), and our results hold for both definitions of social welfare.
    In \Cref{prop:cg-vs-ni}, we show that, for sufficiently differing privileges between groups, recommending $k$ cross-group link recommendations generates higher --- {both utilitarian and Rawlsian} --- social welfare, for all equilibria, than providing no recommendations. 
    In \Cref{prop:crossgroupgood}, we show that, keeping $k$ constant, it is always better to recommend a larger fraction of cross-group link recommendations than a smaller fraction for all equilibria, provided the two fractions are sufficiently separated.\footnote{A gap is needed between fractions of cross-group recommendations because of the multiplicity of equilibria: the best equilibrium for a particular setting of $\rho$ may be better than the worst equilibrium for $\rho + \delta$ where $\delta$ is a small constant.
}

    We first formally state \Cref{prop:cg-vs-ni}.  Intuitively, this proposition says that any equilibrium $E_k\in\cE_k$ induced by all cross-group edge recommendations achieves greater social welfare than any equilibrium $E_0\in\cE_0$ induced by a platform that makes no recommendations (where $\cE_k$ is the set of equilibria of the game with $k$ cross-group recommendations per person).

    \begin{proposition}
    \label{prop:cg-vs-ni}
        For any $k > 0$, there exist constants {$\bar{g}_0, \munderbar{b}_0\in (0,1), \munderbar{n} > 0$} such that for all $n \geq \munderbar{n}$ and opportunity distributions $g, b$ with $g_0 \leq \bar{g}_0$ and $b_0 \geq \munderbar{b}_0$, the following holds. There exist constants $\munderbar{\gamma}, \bar{\gamma}$ with $\munderbar{\gamma} < \bar{\gamma}$ such that for all edge costs $\gamma\in(\munderbar{\gamma},\bar \gamma)$ 
        and all equilibria $E_k\in\cE_k,E_0\in\cE_0$,
        \begin{align}
            \sum_{i \in [n]} u_i(E_k)  & >
             \sum_{i \in [n]} u_i(E_{0}), \text{ and}\label{eq:sw_k_utilitarian}\\
            \min_{i \in [n]} u_i(E_k) & >
             \min_{i \in [n]} u_i(E_{0}). \label{eq:sw_k_rawls}
        \end{align}
    \end{proposition}

    %
    Intuitively, the result tells us that an opportunity has a lower chance of being wasted if there are many edges from high privilege to low privilege individuals.
    The first inequality concerns utilitarian social welfare, and the second concerns Rawlsian social welfare.
    We will sometimes refer to the utilitarian social welfare as the \textit{average endogenous utility} to contrast it with the \textit{average exogenous utility}, which is simply the sum of utilities of individuals in the absence of the platform: $\sum_{i = 1}^n (1 - p_{i0})$.
    Note that because we assume $\rho \abs{B}/\abs{G} \leq 1$, the result requires $\abs{B} \leq \abs{G}$.
    %
    %

    One reason why this result is surprising is that, in general, decreasing connection costs is not always a good thing for the population.
    Indeed, the focus of many of the results in \citet{dwork_equilibria_2024} is that cheaper connectivity can lead to more congestion in the routing of opportunities leading to \textit{lower overall social welfare}.
    This is because connections can induce negative externalities on the neighbors of the pair forming the connection: any opportunity sent across a particular edge is an opportunity \textit{not} sent to another neighbor.
    In general, since individuals do not pay for the negative externalities of their connections, there are more connections between individuals than is socially optimal, leading to more congestion than optimal in the transfer of opportunities.
    Thus, the result tells us that the negative externalities induced by extra connectivity in the network are outweighed by the benefits they bring by channeling opportunities from privileged to unprivileged individuals. 

    The proofs for the utilitarian versus Rawlsian social welfare are slightly different.
    To prove the statement for utilitarian social welfare, we appeal to \Cref{lem:reciprocity-w-recs} to show that almost all individuals form connections with others who have a similar probability of passing each other opportunities.
    This allows us to derive upper and lower bounds on the degrees of all but a constant number of individuals for a fixed set of parameters. 
    This, in turn, allows us to derive upper and lower bounds on the \textit{utilities} of all but a constant number of individuals (since their utility is determined by their neighbors' opportunity distributions and degrees which \Cref{lem:reciprocity-w-recs} greatly simplifies).
    We can then compare the upper bound for all $E_0 \in \cE_0$ with a lower bound for all $E_k \in \cE_k$ and show that the latter is greater than the former.
    For the Rawlsian social welfare, \Cref{lem:reciprocity-w-recs} is not useful, since the properties implied by the lemma may not hold for some number of individuals, and one of these individuals could achieve the minimum. Thus, we argue that \textit{all} members of the unprivileged group form connections, that the lowest-utility member of the unprivileged group has lower utility than that of the privileged group, and finally that the lowest-utility member of the unprivileged group under $k$ cross-group recommendations has utility greater than the corresponding individual under no recommendations.
    
    \begin{figure}
        \centering
        \includegraphics[width=\linewidth]{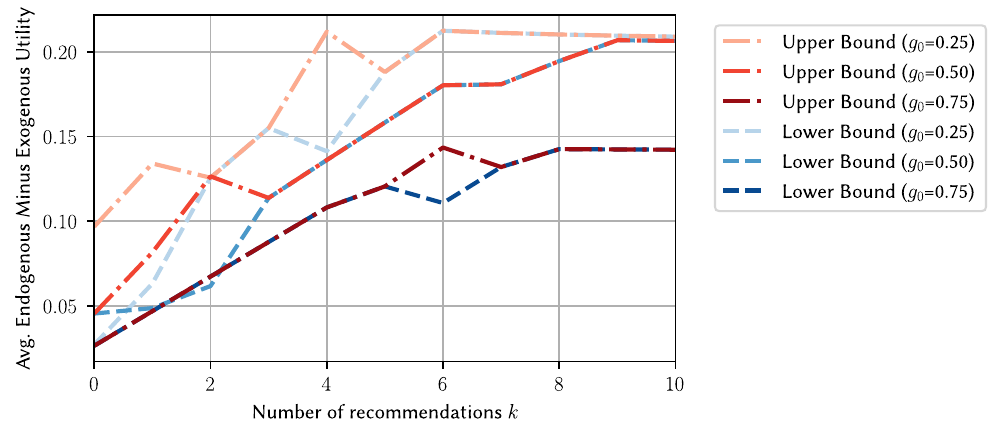}
        \caption{We plot average utility against the number of cross-group recommendations made. It is always better to make $k > 0$ cross-group recommendations than to not make any recommendations. We set $b_0$ so that the exogenous utility ratio is always $2$. We set $\gamma = 0.02$ and $\rho = 1$.}
        \label{fig:utility_bounds_k}
    \end{figure}

    In \Cref{fig:utility_bounds_k}, we plot upper and lower bounds on average utilities for different numbers of recommendations.
    As in \cref{fig:inequality_ratio_comparison,fig:inequality_ratio_comparison_k0}, upper bounds are in red and lower bounds are in blue, and different shades of the same color indicate different values of $g_0$.
    However, unlike in previous figures, on the vertical axis of \cref{fig:utility_bounds_k}, plot the difference between the average endogenous ($\sum_{i=1}^n u_i(E)$) and exogenous ($\sum_{i=1}^n (1 - p_{i0})$  utilities.
    We see that as the number of recommendations increases, the lower bound on average utility increases so that it is above the upper bound at $k=0$. (I.e., for a given $g_0$, the red line at $k=0$ is below the blue line at $k > 0$.)
    However, in general, increasing the number of recommendations --- of any fraction cross-group --- is not always social welfare increasing, since increased connectivity can lead to more inefficiency.
    We can see this by the fact that for values of $g_0$ where the upper and lower bounds are tight towards the right of the plot, the average endogenous minus exogenous utility is decreasing in $k$ past a certain point.

    We next state our second result of this section, which concerns the comparison between different fractions of cross-group recommendations.
    It states that a higher fraction always achieves greater utility, for all equilibria, than a lower fraction, as long as the gap between the lower and higher fractions is sufficiently large. 
    Given a positive integer $k$ and $\rho\in [0,1]$, we use $\cE_{\rho}$ to denote the set of all equilibria induced by any $k$ edge recommendations per person among which $\rho k$ recommendations per person are cross-group.

    
    \begin{proposition} 
    \label{prop:crossgroupgood}
            For any $k > 0$, there exist constants {$\bar{g}_0, \munderbar{b}_0\in (0,1), \munderbar{n} > 0$} such that for all $n \geq \munderbar{n}$ and opportunity distributions $g, b$ with $g_0 \leq \bar{g}_0$ and $b_0 \geq \munderbar{b}_0$, the following holds.
            There exist $\munderbar{\gamma}, \bar \gamma$ where $\munderbar{\gamma} < \bar \gamma$ such that for all edge cost $\gamma \in (\munderbar{\gamma}, \bar \gamma)$ and all $\rho\in (0,1]$, there exists $\delta \in (0, \rho]$ such that for all $\rho' \in \{0, 1/k, \dots, \floor{(\rho - \delta)k}/k \}$ and all equilibria $E_\rho \in \cE_{\rho}, E_{\rho'} \in \cE_{\rho'}$, 
        \begin{align*}
            \sum_{i \in [n]} u_i(E_\rho) & >
             \sum_{i \in [n]} u_i(E_{\rho'}), \text{ and}\\
            \min_{i \in [n]} u_i(E_\rho) & >
             \min_{i \in [n]} u_i(E_{\rho'}).
        \end{align*}
        %
    \end{proposition}
    In other words, the proposition says that any equilibrium induced by all cross-group edge recommendations achieves greater social welfare than any equilibrium induced by a $\rho$ fraction of cross-group recommendations.
    The proof of \Cref{prop:crossgroupgood} appeals to the same upper and lower bounds on utilities used in the proof of the previous result.

    %

    \begin{figure}
        \centering
        \includegraphics[width=\linewidth]{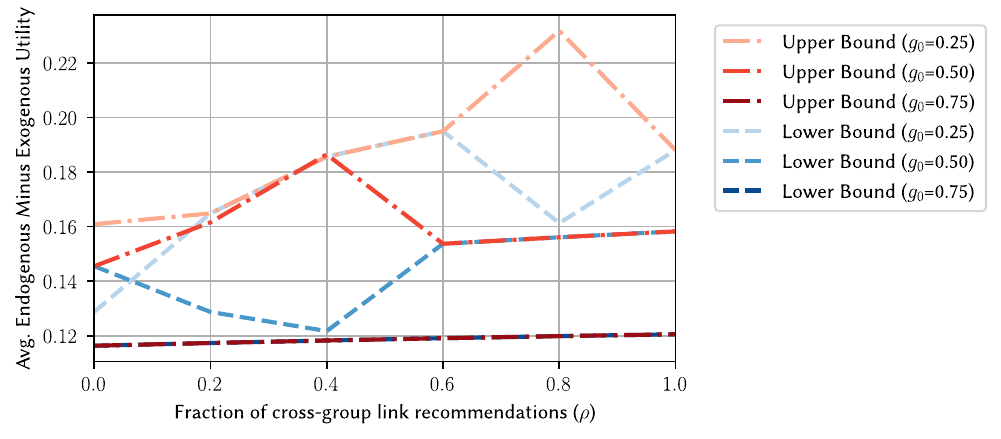}
        \caption{We plot upper (in red) and lower bounds (in blue) on average utility $n^{-1} \sum_{i=1}^n u_i(E)$ minus exogenous utility ($\sum_{i=1}^n(1 - p_{i0})$) against the fraction of cross-group link recommendations $\rho$, keeping $k=5$ constant. We set $b_0$ so that the exogenous utility ratio is always $2$. We set $\gamma = 0.02$. }
        \label{fig:utility_bounds_rho}
    \end{figure}

    In \Cref{fig:utility_bounds_rho}, we show upper and lower bounds on the difference between average endogenous utility and average exogenous utility as $\rho$ varies.
    The setting is the same as in \cref{fig:utility_bounds_k}, except that we fix $k=5$ and vary the proportion of cross-group link recommendations $\rho$.
    Note that since $k=5$, the only valid values of $\rho$ are multiples of $1/5$, the ticks on the horizontal axis.
    The average endogenous minus exogenous utility is increasing, in general, as the fraction of cross-group link recommendations increases, although there are some values of $\rho$ where the change is indeterminate (where the upper and lower bounds are too far apart).
    These parts of the parameter space, e.g., the bounds when $\rho \leq 0.4$ and $g_0 = 0.5$ demonstrate that inequality was \textit{not} sufficiently large so that for every tick on the horizontal axis greater than 0, the blue line at the tick is above the red line for some tick to its left.   
    (By contrast, if we set the exogenous utility ratio to be sufficiently large, then \cref{prop:crossgroupgood} tells us that this would occur.)
    Also, note that when $g_0$ is higher (at darker shades), the benefits of increasing $\rho$ are more modest.
    This is because there is a larger chance that members of $G$ receive no opportunities, so integrating the network can only go so far in increasing the prospects of members of $B$.

\section{Discussion and future work} \label{sec:conclusion}

    In this work, we explore the effect of platform link recommendation policies on inequality and inefficiency.
    We show that facially neutral platforms, which do not intervene on the link recommendation process, increase inequality if inequality is large enough, for all equilibria.
    We also show that cross-group link recommendations decrease inequality and increase efficiency in terms of the total fraction of opportunity used.
    Our work has a number of natural extensions that would be interesting to pursue in future work.


    \paragraph{Search costs.} In our model, we assume individuals pay a fixed cost for each connection.
    Another reasonable model would consider situations in which individuals pay a fixed cost to \textit{search} for a connection.
    In this version, an individual who pays the search cost might receive a random potential connection and then be able to decide whether to connect or not.
    This would push costs further upstream and perhaps induce different equilibria.

    \paragraph{Asymmetric connection costs.} In many situations, there may be asymmetric costs for connections, where the individual proposing the connection pays the cost (perhaps viewed as the cost of proposing, going out on a limb, or suffering the possibility of rejection).
    In this case, members of different groups might propose for a different fraction of their connections on average, perhaps leading to changes in inequality or efficiency.
    This version of the model could also lead to free-riding where a person who would benefit from a connection (even after paying the connection cost) might not propose if they know that the other person in the connection would.

    \paragraph{More than two groups.} Ultimately, real-world platforms consist of more than two levels of exogenous privilege, and it would be interesting to study $m > 2$ groups, each with identical distributions. If we took $k \to \infty$, we could induce a probability distribution over exogenous opportunity distributions.
    This would allow for more precise quantification of inequality calibrated to real-world populations in our model, but would also require new notions of inequality, since ours relies on the fact that there are just two groups.

    \paragraph{Heterogeneous opportunities.} It would also be interesting to model opportunities as (possibly binary) vectors to model complementarities between connections.
    For example, plumbers and appliance repair technicians may encounter many different types of home repair jobs but each only be able to capitalize on certain types of jobs.
    Then utility for each individual could be defined by some function of their vector of opportunities.
    In the home repair example above, utility might be defined for a plumber by whether they received any plumbing opportunities and for an appliance repair technician by whether they received any repair requests (but a plumber might not receive utility for appliance repair requests they receive and similarly for the appliance repair technician).
    Some individuals may be able to make use of some types of opportunities or only derive utility if they receive multiple, complementary opportunities.
    For example, a PhD student might only achieve success if they complete multiple, diverse different projects, deriving diverse ideas from their network.
    This is extension especially relevant under the information transmission interpretation of the model.

    \paragraph{Endogenous opportunity creation.} It is also worth considering how connections between individuals may \textit{create} opportunities. 
    Certain situations might be generative, where the intrinsic fact that two people connect creates opportunities, like a pair of researchers with good chemistry who together generate great project ideas.
    A version of the model that captures the generative aspects of networking might have further implications for the types of network structures that are desirable.

    \paragraph{Noisy or imperfect information of exogenous opportunity distributions.} Finally, we note that platforms and individuals may not have full information about individuals' opportunity distributions (or group memberships). 
    Thus, they may need to act on predictions about individuals, which could lead to effects on inequality or efficiency. \\ \\
    We hope our work inspires further exploration of inequality in professional networking, the role of platforms, and how recommendation systems may have important implications in society.

\pagebreak
    
\bibliographystyle{apalike}
\bibliography{references,manual_references}

\pagebreak

\appendix

\section{Preliminaries} \label{app:preliminaries}

    In this appendix, we include proofs for the results stated in the paper and establish several technical lemmas that are key to our results. 
    
    We first show how to construct recommendations satisfying $k$ recommendations per person and $\rho$ cross-group connections for the unprivileged group, thereby proving existence of such policies as required by many of our proofs.
    First, observe that one can separately construct the sets of cross-group and same-group recommendations.
    For the cross-group recommendations, we can apply simple extensions of Hall's marriage theorem for a complete bipartite graph to show that such a set of cross-group recommendations exists.\footnote{Hall's marriage theorem states conditions for which a bipartite graph has a perfect matching \citep{hall_representatives_1935}. We can generalize to our setting by creating a graph with $\rho k$ nodes associated with each individual in the unprivileged group and $\rho \abs{B}/\abs{G}$ nodes for each member of the privileged group and creating a complete bipartite graph (between nodes in ${B}$ and $G$) of potential recommendations on the node sets. We can verify that the conditions in the theorem hold (the neighborhood of any subset of $G$ is all of $B$ by the completeness of the graph), so that there exists a perfect matching on the graph. Then, cross-group recommendations can be made between pairs of individuals who are associated with matched nodes.} 
    Finding such a set of cross-group recommendations then amounts to reduction to a maximum flow problem using standard arguments.
    Such sets always exist for $\rho k \abs{B} \leq \abs{G}$.
    To construct the same-group recommendations, for each group, one can construct a regular graph of the appropriate degree using standard arguments (\textit{e.g.}, using circulant graphs \citep{sethuraman_graph_2004}).
    Necessary and sufficient conditions for such graphs to exist are that $\abs{G} (1-\rho) k$ and $\abs{G} (1-\rho \abs{B}/\abs{G}) k$ are even. Throughout, for simplicity, we will assume such conditions hold, but our analysis could be simply extended to the cases when they don't hold.

    Next we state the following implications of the equilibrium conditions in \cref{eq:noadd,eq:nodelete} used throughout the proofs.
    For \cref{eq:nodelete} and $i \in [n]$, \cref{eq:utilformula} and \cref{eq:nodelete} together imply 
    \begin{align}
        0 &\geq u_i(E \setminus \{ (i,j)\}) - u_i(E) && \forall (i,j) \in E \nonumber \\
        \implies 0&\geq \gamma - \min_{\ell \in N_i(E) \setminus Q_i} p_{i0}\frac{ \mu_{\ell}(d_\ell)}{d_\ell} \prod_{\substack{j \in N_i(E) \setminus Q_i \\ j \neq \ell}} \paren{1 - \frac{\mu_{j}(d_j)}{d_j}}\prod_{{j \in N_i(E) \cap Q_i}} \paren{1 - \frac{\mu_{j}(d_j)}{d_j}} \nonumber \\
        &\geq \gamma - \min_{\ell \in N_i(E) \setminus Q_i} p_{i0}\frac{ \mu_{\ell}(d_\ell)}{d_\ell}  \paren{1 - \frac{\mu_{\ell}(d_\ell)}{d_\ell}}^{d_i-k-1} \prod_{{j \in N_i(E) \cap Q_i}} \paren{1 - \frac{\mu_{j}(d_j)}{d_j}}. \label{eq:nosever}
    \end{align}
    The implication follows from the fact that $i$ must not be willing to sever their least valuable connection.
    The second inequality after the implication follows from the fact that for $\ell$ achieving the minimum and all $j$ such that $j \in N_i(E) \setminus Q_i$, it holds $\mu_{j}(d_j)/d_j \geq \mu_{\ell}(d_\ell)/d_\ell$.

    We can similarly write 
    \begin{align}
        0 &\geq \min_{\ell \in \{i, j \}} u_\ell(E \cup \{ (i,j)\}) - u_\ell(E)  \nonumber
    \end{align}
    where the first line is implied by \cref{eq:noadd}. Suppose the minimum is achieved at $\ell = i$. Then, the equilibrium condition can be written
    \begin{align}
        0 &\geq p_{i0} \frac{\mu_j(d_j+1)}{d_j + 1} \prod_{\ell \in N_i(E)} \paren{1 - \frac{\mu_\ell(d_\ell)}{d_\ell}} - \gamma && \forall (i,j) \not\in E. \label{eq:nocreate}
    \end{align}
    In other words, the marginal benefit of connecting to $j$ must not be worth the cost.
    It is written $d_j + 1$ because forming the connection would increase $j$'s degree by 1.

    \paragraph{Off-platform versus on-platform links.} Part of any reasonable definition of exogenous privilege will include the networks that privileged individuals form as a result of, essentially, their birth. 
    Individuals may have family or friend connections that are essentially guaranteed to form regardless of the existance of the platform.
    These are not the types of links that we are interested in modeling: instead we are interested in the kinds of connections that are directly facilitated by a professional networking platform, and might not otherwise exist.
    The model could be naturally extended to account for ``off-platform'' or ``exogenous'' links that are already formed when individuals arrive to the platform.
    In this case, the definition of exogenous privilege could be extended to account for these built-in connections, and the distribution $p_{i0}$ could be re-defined to mean the probability that an individual does not receive an exogenous opportunity \textit{or} an opportunity from one of their exogenous connections, and $\mu_{i}(d_i)$ throughout could be redefined to mean the expected number of opportunities an individual passes to their \textit{endogenous} connections.
    The results in this paper could then be extended to account for these built-in links by ensuring that equilibria still form. 
    If every member of the privileged group has the same number of connections to privileged and unprivileged individuals, the results in the paper follow directly by propagating through the new definitions of $p_i$ and $\mu_i(\cdot)$.

\section{Deferred proofs}

\subsection{Preliminary lemmata}

    In this subsection, we introduce several lemmas that will be useful in proving the rest of our results.

    \begin{lemma}\label{lem:no-blue-edges}
        For all $n, g, \rho, k, \gamma$, and $b$ such that $ b_0 (1-b_0 - b_1) \leq \gamma$ and $E \in \cE$, for all $j \in B$,
        \begin{align*}
            N_j(E) \setminus Q = \varnothing.
        \end{align*}
    \end{lemma}

    \begin{proof}[Proof of \Cref{lem:no-blue-edges}]
        Note that the probability that an individual does not receive an opportunity is at at most $b_0$ (since $g_0 \leq b_0$ by stochastic dominance) and the probability that a member of the blue group receives an opportunity is at most $(1-b_0-b_1)$. 
        Thus, the benefits of an organic connection can be no more than $b_0(1-b_0 - b_1)$ and the cost is $\gamma$, so \cref{eq:nodelete} above implies that no one must be willing to form an organic connection to a member of $B$. 
    \end{proof}

    \begin{lemma}\label{lem:some-green-edges}
        For $k=0$, all privileged group opportunity distributions $g$ and edge costs $\gamma < (0, g_0(1-g_0 - g_1))$, all members of $G$ except at most 1 have at least one connection.
    \end{lemma}

    \begin{proof}[Proof of \Cref{lem:some-green-edges}]
        Note that the fact that $\gamma < g_0(1-g_0 - g_1)$ implies 
        \begin{align*}
            1 - g_0 < 1 - g_0 ( g_0 + g_1 ) - \gamma.
        \end{align*}
        Observe that the left-hand side is the utility of a green-group member with no connections and the right-hand side is the utility of a green-group member with a single connection to another green-group member with only one connection.
    \end{proof}
    
\subsection{Deferred proofs for Section~\ref{sec:equilibrium}}

    \begin{proof}[Proof of \Cref{lem:reciprocity-w-recs}] For the first statement in part (a) of the result, the proof of Proposition 4.9 of \citet{dwork_equilibria_2024} can be extended simply to the setting of \Cref{lem:reciprocity-w-recs}, which accounts for the incentives resulting from recommendations. 
    %
    In fact, the expression for the marginal utilities of adding or removing the edge (\Cref{eq:noadd,eq:nodelete}) are the same.
    Thus, we can follow the same argument as in the proof of Proposition 4.9, reproduced here for the sake of being self-contained. If some pair have the following properties
    \begin{enumerate}
        \item $d_i = d_{i'}$
        \item $i,j$ are members of the same group 
        \item there exists $j, j'$ such that $(i,j), (i',j') \in E$, \textit{and} $(i,j), (i',j') \not\in Q$ where
        \begin{align*}
            \frac{\mu_j(d_j)}{d_j} < \frac{1}{d_i + 1} \left\lfloor \frac{\mu_i(d_i + 1)}{\varepsilon}\right\rfloor \varepsilon \quad\quad \text{and} \quad\quad\frac{\mu_{j'}(d_{j'})}{d_{j'}} < \frac{1}{d_{i'} + 1} \left\lfloor \frac{\mu_{i'}(d_{i'} + 1)}{\varepsilon}\right\rfloor\varepsilon
        \end{align*}
    \end{enumerate}
    then $(i,i')$ must be connected. 
    To see why this is true, suppose for contradiction that they are not connected.
    Then $i,i'$ would rather defect and connect to each other while severing their connections to $j,j'$ since we assumed that the marginal benefits of the $(i,j)$ and $(i',j')$ edge are less than the benefits of severing those connections and creating a $(i,i')$ link.
    But this contradicts the fact that $E$ is an equilibrium.

    The above argument implies that, in an equilibrium, at most $d_i+1$ individuals of the same group as $i$ can exist where those conditions hold, since each pair of such individuals must be connected and thus form a $(d_i+1)$-clique.
    We will upper bound how many of these individuals there can be in the population.

    %
    We first need to upper bound how many $\varepsilon$-grid cell-degree combinations there can be.
    As in the proof of Proposition 4.9, there can only be $\max_{i \in [n]} d_i/\varepsilon$ cells. 
    And since utility must be nonnegative and all recommended connections will form by \Cref{lem:all-edges-form}, we have
    \begin{align*}
        0 \leq u_i(E) = 1 - p_{i0} \prod_{j \in N_i(E)}\paren{1 - \frac{\mu_j(d_j)}{d_j}} - \gamma (d_i - k)
    \end{align*}
    which implies $d_i \leq \gamma^{-1} + k$.
    Thus, there can be no more than $2(\gamma^{-1} + k)\varepsilon^{-1}$ cells and no more than $2(\gamma^{-1} + k)^2\varepsilon^{-1}$ cell-degree combinations.
    Now, since at most $d+1$ individuals of degree $d$ in each cell can have a connection violating \cref{eq:genreciprocity}, this implies at most $2(\gamma^{-1} + k)^2(\gamma^{-1} + k + 1)\varepsilon^{-1}$ individuals across all cells can have a connection violating \cref{eq:genreciprocity}.
    Define $\cV$ to consist of all individuals $i \in [n]$ such that for all $j \in N_i(E)$ it holds
    \begin{align*}
        \frac{\mu_j(d_j)}{d_j} \geq \frac{1}{d_i + 1} \left\lfloor \frac{\mu_i(d_i + 1)}{\varepsilon}\right\rfloor,
    \end{align*}
    and for whom there exists some $j$ such that $(i,j) \not\in E$, $d_i = d_j$ and $i,j$ in the same group.
    Notice all individuals in $\cV$ satisfy the left-hand inequality in \cref{eq:genreciprocity}. 
    Define $S \defeq \{ i \in \cV \; : \; \forall j \in N_i(E), j \in \cV\}$ to consist of all individuals in $\cV$ whose connections are all in $\cV$. 
    Notice for all $i \in S$ and $j \in N_i(E)$ it must hold that
    \begin{align*}
        \frac{\mu_i(d_i)}{d_i} \geq \frac{1}{d_j + 1} \left\lfloor \frac{\mu_j(d_j + 1)}{\varepsilon}\right\rfloor \varepsilon \geq \frac{\mu_j(d_j + 1)}{d_j+1} - \varepsilon,
    \end{align*}
    which is the right-hand inequality in \cref{eq:genreciprocity}. Since the number of individuals in $[n] \setminus \cV$ is upper bounded by $2(\gamma^{-1} + k)^2(\gamma^{-1} + k + 1)\varepsilon^{-1}$ and each individual's degree is upper bounded by $\gamma^{-1} + k$, the number of individuals connected to individuals in $[n] \setminus \cV$ must be bounded by $2(\gamma^{-1} + k)^2(\gamma^{-1} + k + 1)^2\varepsilon^{-1}$.
    Finally, since all $i \in S$ satisfy \cref{eq:genreciprocity}, we have the first statement in the first part of the result.

    For the second statement in part (a) of the result, assume without loss of generality that $i, j \in G$. (So their opportunity distributions are given by $g$.) 
    Suppose $i \in S$ and $j \in N_i(E)$.
    Note that, by the first statement in part (a), it holds $j \in S$.
    Thus, 
    \begin{align*}
        & \frac{\mu_i(d_i + 1)}{d_i + 1} \leq \frac{\mu_j(d_j)}{d_j} + \varepsilon, \\
        \implies & (d_j - 1) {\mu_i(d_i + 1)} + {\mu_i(d_i + 1)} \leq d_i{\mu_j(d_j)} + {\mu_j(d_j)} + \varepsilon,  \\
        \implies & (d_j - 1) {\mu_i(d_i)} - (d_j - 1)({\mu_i(d_i)} - {\mu_i(d_i + 1)}) + {\mu_i(d_i + 1)} \\
        &\leq d_i{\mu_j(d_j - 1)} - d_i ({\mu_j(d_j - 1)} - {\mu_j(d_j)}) + {\mu_j(d_j)} + \varepsilon, \\
        \implies & \frac{\mu_i(d_i)}{d_i} \leq \frac{\mu_j(d_j - 1)}{d_j - 1} + \frac{{\mu_i(d_i)} - {\mu_i(d_i + 1)}}{d_i}- \frac{{\mu_j(d_j - 1)} - {\mu_j(d_j)}}{d_j - 1} + \frac{\mu_j(d_j) - \mu_i(d_i + 1)}{d_i(d_j - 1)} + \varepsilon,
    \end{align*}
    where each implication follows by rearranging the previous inequality.
    Observe, to prove the statement, we just need to show
    \begin{equation}
    \label{eq:goal-3-1}
        \frac{{\mu_i(d_i)} - {\mu_i(d_i + 1)}}{d_i}- \frac{{\mu_j(d_j - 1)} - {\mu_j(d_j)}}{d_j - 1} + \frac{\mu_j(d_j) - \mu_i(d_i + 1)}{d_i(d_j - 1)} \leq 0.
    \end{equation}
    By our assumption that $i$ and $j$ belong to the same group, we have $\mu_i = \mu_j$. 
    Next, we break the analysis into two cases.

    Case 1: 
    $d_j \le d_i + 1$. 
    Notice,
    \begin{equation}
    \label{eq:part-1-2}
    \frac{ {\mu_i(d_i) - \mu_i(d_i + 1)}}{d_i}  - \frac{ {\mu_j(d_j - 1) - \mu_j(d_j)} }{d_j - 1}  
        = \frac{ \sum_{\ell \geq d_i+1} g_\ell }{d_i}  - \frac{ \sum_{\ell \geq d_j} g_\ell }{d_j - 1} \le 0.
    \end{equation}
    Moreover, since the function $\mu_i = \mu_j$ is non-decreasing and $d_j \le d_i + 1$, we have
    \begin{equation}
    \label{eq:part-3}
    \frac{\mu_j(d_j) - \mu_i(d_i + 1)}{d_i(d_j - 1)} \le 0.
    \end{equation}
    Summing up \eqref{eq:part-1-2} and \eqref{eq:part-3}, we get \eqref{eq:goal-3-1}, as desired.

    Case 2: $d_j > d_i + 1$.
    Notice, by \cref{eq:genreciprocity} and the fact that $j \in S$, we have
    \begin{align*}
        &\frac{\mu_j(d_j)}{d_j} \geq \frac{\mu_i(d_i + 1)}{d_i+1} - \varepsilon \\
        \implies & \sum_{\ell \geq 1} \paren{\frac{\min\{ (d_j, \ell - 1) \}}{d_j} - \frac{\min\{ (d_i+1, \ell - 1) \}}{d_i+1}} g_\ell \geq - \varepsilon \\
        \implies & \sum_{\ell = 1}^{d_i+1} \paren{\frac{  (\ell - 1) }{d_j} - \frac{( \ell - 1) }{d_i+1}} g_\ell  
        + \sum_{\ell = d_i+2}^{d_j} \paren{\frac{ (\ell - 1)}{d_j} - 1} g_\ell\geq - \varepsilon \\
        \implies & \sum_{\ell = 1}^{d_i+1} (\ell - 1)\paren{{ (d_i+1)  } - {{d_j} }} g_\ell  
        + d_i \sum_{\ell = d_i+2}^{d_j} \paren{{ (\ell - 1)} - {d_j}} g_\ell\geq - \varepsilon d_j (d_i+1) \\
        \implies & -\sum_{\ell = 1}^{d_i+1} (\ell - 1) g_\ell  
        - d_i \sum_{\ell = d_i+2}^{d_j}  g_\ell\geq - \varepsilon d_j (d_i+1) \tag{$d_j > d_i + 1$}\\ 
        \implies & \frac{\mu_G(d_i)}{d_j (d_i+1)} \leq \varepsilon  
    \end{align*}
    Where each line except for the second-to-last follows by rearranging.
    But the last line contradicts the assumption that $\varepsilon$ is small.
    Thus, we have a contradiction with \cref{eq:genreciprocity}, and it must not hold that $d_j > d_i + 1$.

    For part (b), note that in our construction of $\cV$ (and therefore $\cS$) above, we required that there exists $(i,j) \not\in E$, such that $d_i = d_j$ and $i,j$ in the same group. Thus, the second part holds by the construction of $\cS$.

    For final statement in the result (about small enough $\varepsilon$), notice that, since each person's degree is bounded by $\gamma^{-1} + k$ 
    %
    the range of $\mu_G(\cdot)$  and $\mu_B(\cdot)$ are finite sets, which means the map $d \mapsto \mu_i(d)/d$ takes values in a finite set. Thus, for small enough $\varepsilon$, the inequalities in \cref{eq:genreciprocity} hold exactly. 
    \end{proof}
    
    \begin{proof}[Proof of \Cref{fact:eqbm-exists}]
        The proof follows the pattern of that of \citet{dwork_equilibria_2024}, Proposition C.1.
        First, note from \Cref{lem:all-edges-form} all recommended connections will form, so each member of $B$ will have at least $k$ connections. Additionally, by \Cref{lem:no-blue-edges}, since we assumed $\gamma > b_0(1-b_0 - b_1)$, no member of $B$ will have more than $k$ connections.
        We will prove the existence of an equilibrium where each individual in $G$ has the same degree as every other member of their group (denote their degree $d_G$).
        %
        First, note that if $$\gamma > g_{0}\frac{ \mu_{G}(k+1)}{k+1} \paren{1 - \frac{\mu_G(k)}{k}}^{k\paren{1-\rho \abs{B}/\abs{G}}} \paren{1 - \frac{\mu_B(k)}{k}}^{k\rho \abs{B}/\abs{G}},$$ then no member of $G$ will form an organic (\textit{i.e.}, not recommended) connection with a member of $G$ since the cost of the connection outweighs the marginal benefit.
        %
        Thus, any $k$-regular graph is an equilibrium.

        Next, suppose $\gamma$ is less than the right-hand side of the inequality above.
        Note that, by the equilibrium conditions \cref{eq:noadd} and \cref{eq:nodelete}, if such an equilibrium exists and $d_G > k$, then
        \begin{align}
            \gamma \leq  g_{0}\frac{ \mu_{G}(d_G)}{d_G}  \paren{1 - \frac{\mu_{G}(d_G)}{d_G}}^{d_G-k\rho \abs{B}/\abs{G}-1} \paren{1 - \frac{\mu_B(k)}{k}}^{k\rho \abs{B}/\abs{G}} \label{eq:gammalower}
        \end{align}
        and 
        \begin{align}
           \gamma \geq  g_{0} \frac{\mu_G(d_G+1)}{d_G  + 1} \paren{1 - \frac{\mu_G(d_G)}{d_G }}^{d_G - k\rho\abs{B}/\abs{G}} \paren{1 - \frac{\mu_B(k)}{k }}^{k\rho\abs{B}/\abs{G}} \label{eq:gammaupper}
        \end{align}
        It is easy to verify that for fixed $d_G$, the right-hand side of \cref{eq:gammalower} is greater than the right-hand side of \cref{eq:gammaupper}.
        Now notice that if we argue that 
        \begin{align}
            \begin{aligned}
            \bigcup_{d_G \geq k+1} &\Bigg[g_{0} \frac{\mu_G(d_G+1)}{d_G  + 1} \paren{1 - \frac{\mu_G(d_G)}{d_G }}^{d_G - k\rho\abs{B}/\abs{G}} \paren{1 - \frac{\mu_B(k)}{k }}^{k\rho\abs{B}/\abs{G}}, \\
            &g_{0}\frac{ \mu_{G}(d_G)}{d_G}  \paren{1 - \frac{\mu_{G}(d_G)}{d_G}}^{d_G-k\rho \abs{B}/\abs{G}-1} \paren{1 - \frac{\mu_B(k)}{k}}^{k\rho \abs{B}/\abs{G}} \Bigg] \\
            = &\Bigg(0, g_{0}\frac{ \mu_{G}(k+1)}{k+1} \paren{1 - \frac{\mu_G(k)}{k}}^{k\paren{1-\rho \abs{B}/\abs{G}}} \paren{1 - \frac{\mu_B(k)}{k}}^{k\rho \abs{B}/\abs{G}}\Bigg] 
            \end{aligned}\label{eq:unionofintervals}
        \end{align}
        then the result is proved, since this implies that for all $\gamma >  b_0(1 - b_0 - b_1)$, there exists a $d_G$ such that \cref{eq:gammalower} and \cref{eq:gammaupper} are satisfied.
        %
        First, observe that 
        \begin{align*}
            \frac{ \mu_{G}(d_G + 1)}{d_G + 1}  \paren{1 - \frac{\mu_{G}(d_G + 1)}{d_G + 1}}^{d_G-k\rho \abs{B}/\abs{G}} > \frac{\mu_G(d_G+1)}{d_G  + 1} \paren{1 - \frac{\mu_G(d_G)}{d_G }}^{d_G - k\rho\abs{B}/\abs{G}}.
        \end{align*}
        Thus, the $(d_G+1)$th interval overlaps with the $d_G$th.
        Second, note that the right endpoint of the $d_G = k+1$ interval is the right endpoint of the right-hand interval in \cref{eq:unionofintervals}.
        Finally, note that the limit as $d_G \to \infty$ of the left endpoint of the intervals is $0$.
        Thus, \cref{eq:unionofintervals} holds.

        Finally, as long as $n$ is large enough that it is possible to construct a graph where all members of $B$ have degree $k$ and all members of $G$ have degree $d_G$, the result holds.
        \end{proof}

    \subsection{Deferred proofs for Section~\ref{sec:inequality}}

    \begin{proof}[Proof of \Cref{prop:ineqworse}]
    From \Cref{lem:some-green-edges}, note that all but at most one green group member will have at least one connection.
    Thus, for all but at most one green-group member $i \in G$
    \begin{align}
        u_i(E) \geq 1 - g_0(g_0 + g_1) - \gamma > 1 - g_0. \label{eq:geqg0}
    \end{align}
    This implies
    \begin{align*}
        {\abs{G}}^{-1}\sum_{i \in G} {u_i(E)} \geq \paren{1 - \frac{1}{\abs{G}}} \paren{1 - g_0(g_0 + g_1) - \gamma} 
    \end{align*}
    and so for $\abs{G} > (1 - g_0(g_0 + g_1) - \gamma ) / (g_0 (1 - (g_0 + g_1)) - \gamma)$, 
    \begin{align*}
        {\abs{G}}^{-1}\sum_{i \in G} {u_i(E)} > 1 - g_0.
    \end{align*}

    Next, we will argue that members of the blue group will not form any connections, so that each of their utilities are equal to $1-b_0$. 
    To do so, we will argue that no one (from either group) will want to form a connection with a blue group member (for $b_0$ sufficiently large).
    Suppose $i \in [n]$ has at least one connection to a blue group member in an equilibrium edge set $E$, and
    let $\ell$ denote such a blue group member.
   \remove{ \cd{Suspect first line should be:
      \begin{align*}
         u_i(E) - u_i(E \setminus \{ (i,\ell) \}) &= \gamma - p_{i0} \frac{\mu_\ell(d_\ell)}{d_\ell}\prod_{j \; : \; (i,j) \in E \setminus \{ (i,\ell)\}}\paren{1 - \frac{\mu_j(d_j)}{d_j}}  
     \end{align*}
    Was $\mu_\ell(d_\ell)$ defined?  In our case, is it $\mu_B(d_\ell)$ since $\ell \in B$?  Similarly for $\mu_j(d_j)$?
 }
    \ch{Fixed typos and defined $\mu_\ell$ above; it shouldn't be $\mu_B(d_j)$ since we aren't assuming group membership of $j$.}
    }
    Then
    \begin{align*}
        u_i(E) - u_i(E \setminus \{ (i,\ell) \}) &= \gamma - p_{i0} \frac{\mu_B(d_\ell)}{d_\ell}\prod_{j \; : \; (i,j) \in E \setminus \{ (i,\ell)\}}\paren{1 - \frac{\mu_j(d_j)}{d_j}}  \\
        &\geq \gamma - \frac{\mu_B(d_\ell)}{d_\ell} \\
        &\geq \gamma - {\sum_{j \geq 1} (j - 1) b_j} \\
    \end{align*}
    where the first line is by plugging in the expressions for $u_i$, the second is by the fact that $p_{i0}$ and the product term $\prod_{j \; : \; (i,j) \in E \setminus \{ (i,\ell)\}}\paren{1 - \frac{\mu_j(d_j)}{d_j}}$ are both bounded above by 1 and the third line is because 
    \begin{align*}
        \mu_\ell(d) = \sum_{j \geq 1} \min \curly{j - 1 , d} b_j\leq \sum_{j \geq 1}(j - 1) b_j.
    \end{align*}
    Thus, if $\sum_{j \geq 1}(j - 1) b_j \leq \gamma$, then no individual will want to connect to a blue group member.

    Finally, these two facts imply
    \begin{align*}
        \eur(E) &= \frac{{\abs{G}}^{-1}\sum_{i \in G} {u_i(E)}}{{\abs{B}}^{-1} \sum_{i \in B} {u_i(E)}} \\
        &= \frac{{\abs{G}}^{-1}\sum_{i \in G} {u_i(E)}}{1-b_0} \\
        &> \frac{1-g_0}{1-b_0}
    \end{align*}
    where the last line is implied by \cref{eq:geqg0} for all $\abs{G}$ sufficiently large.
    \end{proof}

\begin{proof}[Proof of \Cref{prop:cg_ineq}]
        Let $\munderbar{b}_0 = {g_0}/{(1 - (1-g_0)(1 - \gamma \mu_G(\gamma^{-1}))^k)}$.
        Fix any equilibrium edge set $E$ for any $n \in \N$, $b$, $\gamma$ and $b_0 > b_0'$.
        Define $\dgupper$ to be an upper bound on $d_j$ for all $j \in G$ with respect to $E$.
        Observe that
        \begin{align*}
            \eur(E) \leq \frac{1}{\min_{j \in B} u_j(E)}
        \end{align*}
        since $u_i(E) \leq 1$ for all $i \in [n]$ so $\abs{G}^{-1} \sum_{i \in G} u_i(E) \leq 1$.
        Also, since all recommended edges form by \Cref{lem:all-edges-form}, 
        it holds
       \begin{align*}
            \min_{j \in B} u_j(E) &\geq 1-b_0\paren{1-\frac{\mu_G(\dgupper)}{\dgupper}}^k \\
            &\geq 1 - b_0 \paren{1 - {\gamma \mu_G(\gamma^{-1})}}^k
       \end{align*} 
       where the first line comes from the fact that $\mu_j(d) / d$ is non-increasing in $d$ the second comes from the fact that $\dgupper \leq \gamma^{-1} + k$ (since utilities are bounded above by 1 and each edge costs $\gamma$, having more than $\gamma^{-1}$ non-recommended connections will always yield negative utility).
        Thus, by assumption 
        \newcommand{\gammuk}{\paren{1 - {\gamma \mu_G(\gamma^{-1})}}^k}
        \begin{align*}
            &b_0 > \frac{g_0}{1 - (1 -g_0) \gammuk} \\
            \implies & b_0 (1 - (1-g_0)\gammuk) > g_0 \\
            \implies & 1 - g_0 - b_0 (1 - g_0) \gammuk > 1 - b_0 \\
            \implies &1 - b_0 < (1 - g_0)(1 - b_0 \gammuk) \\
            \implies & \frac{1}{1 - b_0 \gammuk} < \frac{1 - g_0}{1 - b_0} = \eur(\varnothing)
        \end{align*}
        by rearranging. Summarizing the analysis above, we have
        \[
        \eur(E) \le \frac 1{\min_{j \in B} u_j(E)} \le \frac 1{1 - b_0 \paren{1 - {\gamma \mu_G(\gamma^{-1})}}^k} < \eur(\varnothing).
        \]
        It remains to prove $\munderbar{b}_0 < 1$. Indeed,
        \begin{align*}
            \munderbar{b}_0 = \frac{g_0}{1 - (1-g_0)(1 - \gamma \mu_G(\gamma^{-1}))^k} < 1
        \end{align*}
        by the fact that $\gamma \mu_G(\gamma^{-1})) \in (0, 1)$ and rearranging.
    \end{proof}
\subsection{Deferred proofs for Section~\ref{sec:welfare}}    
    \begin{lemma}\label{lem:utilbounds}
        For any constant $\varepsilon > 0$ and any $k > 0$, there exist constants $\munderbar{b}_0 \in (0, 1), \munderbar{n}$ such that for all $n > \munderbar{n}$ and distributions $g, b$ with $b_0 \geq \munderbar{b}_0$, the following holds.
        There exists a constant $\munderbar{\gamma}, \bar \gamma$ with $\munderbar{\gamma} < \bar \gamma$ such that for all $\gamma \in (0, \bar \gamma)$, all $\rho \in [0, 1]$ and any equilibrium edge set $E \in \cE$, there exists a set $S \subseteq [n]$ such that $\abs{S} \geq n - O(1)$ and
        \begin{align}
            {u_i}(E) &\geq 1 - b_0 \paren{1 - \frac{\gamma }{g_0 e^{-\mu_G} + \gamma \mu_G^{-1}}}^{\rho k} - \varepsilon && \forall i \in S \cap B 
            \label{eq:uBlower}\\
            {u}_i(E) &\leq 
               1 - b_0 \paren{1 - \frac{\gamma}{ g_0 e^{-\mu_G}  - \gamma \mu_G^{-1}}}^{\rho k} + \varepsilon && \forall i \in S \cap B
            \label{eq:uBupper}\\
            {u}_i(E) &\geq 1 - g_0 (1 + \mu_G) e^{-\mu_G} - \varepsilon && \forall i \in S \cap G\\ 
            {u_i}(E) &\leq 1 -  g_0(1 + \mu_G) e^{-\mu_G} + \gamma (k + 1) + \varepsilon && \forall i \in S \cap G  
        \end{align}
    \end{lemma}
    \remove{
    \begin{lemma}\label{lem:utilbounds}
        For any privileged-group exogenous opportunity distribution $g$, there exists a constant $\bar \gamma$ such that for any edge cost $\gamma \in (0, \bar \gamma)$, fraction of cross-group link recommendations for the unprivileged group $\rho$ and constant $\varepsilon > 0$, there exists $\munderbar{b}_0$ such that for all $n > k$, $b$ such that $\mu_B < 1$ and $b_0 > \munderbar{b}_0$ and for any equilibrium edge set $E \in \cE$, there exists a set $S \subseteq [n]$ such that $\abs{S} \geq n - O(1)$ and
        \begin{align}
            {u_i}(E) &\geq 1 - b_0 \paren{1 - \frac{\gamma }{g_0 e^{-\mu_G} + \gamma \mu_G^{-1}}}^{\rho k} - \varepsilon && \forall i \in S \cap B 
            \label{eq:uBlower}\\
            {u}_i(E) &\leq 
               1 - b_0 \paren{1 - \frac{\gamma}{ g_0 e^{-\mu_G}  - \gamma \mu_G^{-1}}}^{\rho k} + \varepsilon && \forall i \in S \cap B
            \label{eq:uBupper}\\
            {u}_i(E) &\geq 1 - g_0 (1 + \mu_G) e^{-\mu_G} - \varepsilon && \forall i \in S \cap G\\ 
            {u_i}(E) &\leq 1 -  g_0(1 + \mu_G) e^{-\mu_G} + \gamma (k + 1) + \varepsilon && \forall i \in S \cap G  
        \end{align}
    \end{lemma}
    }

\begin{proof}[Proof of \Cref{lem:utilbounds}]
        Our proof consists of the following steps:
        \begin{enumerate}
            \item Derive upper and lower bounds on the degrees of each demographic group using the equilibrium conditions and the conditions in \cref{eq:genreciprocity}.
            \item Use these degree bounds to derive upper and lower bounds on the utility of individuals in each group.
        \end{enumerate}
        Throughout the proof, we will make use of the constant $\varepsilon > 0$ to prove that, if $b_0$ is sufficiently large and $\gamma$ is sufficiently small, then the bounds on degrees and utilities are considerably simpler. 
        We will use $\varepsilon$ throughout, and assume that $\bar \gamma$ and $\munderbar{b}_0$ are set so that the inequalities are satisfied.
        The $\varepsilon$ on sequential lines will not necessarily be the same; we may use the $\varepsilon$ on one line to bound those those on subsequent lines.
        We will liberally make use of the fact that we can set $\varepsilon$ in each line of the proof so that the result holds for the $\varepsilon$ in the statement of the result.
        
    Step 1. 
    Degree bounds for members of $B$ are established as follows. Set $\munderbar{\gamma} = b_0(1-b_0-b_1)$. From \Cref{lem:no-blue-edges}, we have for all $i \in B$, $N_i(E) \setminus Q = \varnothing$.
    With \Cref{lem:all-edges-form}, this implies that $d_i = k$.
    
    Next, we establish degree bounds for members of $G$.
    To do so, we will derive explicit forms for the equilibrium conditions \cref{eq:noadd} and \cref{eq:nodelete}.
    Then we will use these conditions to construct maximum and minimum degree bounds for members of $G \cap S$, since individuals with the maximum and minimum degrees must also satisfy these equilibrium conditions.
    
    We first derive a refined version of \cref{eq:nosever} for $i \in G \cap S$. 
    For a sufficiently small global constant $\zeta > 0$, let $S$ be a set of cardinality $n - C(\gamma, k)/\zeta$ satisfying the conditions in \Cref{lem:reciprocity-w-recs} so that, for $i \in S$, $j \in N_i(E)$
    \begin{align}
        \frac{\mu_{j}(d_j + 1)}{d_j + 1} \leq \frac{\mu_{i}(d_i)}{d_i}. \label{eq:genreciprocity2}
    \end{align}
    and $j \in S$.
    (That is, we set $\zeta$ small enough so that \cref{eq:genreciprocity} holds exactly.)
    %
    %
    For all $i \in G \cap S$, from \Cref{lem:no-blue-edges}, it holds $N_i(E)\subseteq G$. 
    Thus, the second statement in \cref{lem:reciprocity-w-recs} part (a) holds. I.e., for all $i \in S$ and all $\ell \in N_i(E)$, it holds
    \begin{align*}
        \frac{\mu_G(d_\ell)}{d_\ell} \in \sqparen{ \frac{\mu_G(d_i - 1)}{d_i - 1}, \frac{\mu_G(d_i + 1)}{d_i + 1} }.
    \end{align*}
    Thus, for all $i \in S$,  \cref{eq:nosever} implies, for $d_i \geq 2$,
    \begin{align}
        0 &\geq \gamma - \max_{\frac{\mu_G(d_i + 1)}{d_i + 1} \leq v \leq  \frac{\mu_G(d_i - 1)}{d_i - 1} } {g_0 v} \paren{1 - {v}}^{d_i - k - 1} \prod_{j \; : \; (i,j) \in Q} \paren{1 - \frac{\mu_{j}(d_j)}{d_j}  }.\label{eq:nosever_g}
    \end{align}
    
    %

    We next derive a refined version of \cref{eq:nocreate} for $i \in G \cap S$.
    From part (b) of \Cref{lem:reciprocity-w-recs} (and the fact that $\zeta$ is sufficiently small), there must exist some $j$ such that $d_j = d_i$ and $i,j \in G$ but $(i,j) \not\in E$. 
    %
    Thus,
    \begin{align}
        &0 \geq \frac{g_0 \mu_G(d_i+1)}{d_i + 1} \paren{1 - \frac{\mu_G(d_i+1)}{d_i+1}}^{d_i - k} \min_{a \in \{i, j\}} {\prod_{\ell \; : \; (a,\ell) \in Q} \paren{1 - \frac{\mu_{\ell}(d_\ell)}{d_{\ell}}}} - \gamma, \label{eq:noconnect_g}
    \end{align}
    since $i$ or $j$ must not be willing to connect to the other.

    Now, suppose there exist constants $\munderbar d_G, \bar d_G \in \N$ (depending on $b, g, \gamma, \rho$ but not $n$) such that $d_i \in [\munderbar d_G, \bar d_G]$ for all $i \in G \cap S$.
    %
    %
    We next prove that such $\dgupper$ and $\dglower$ exist and establish bounds on them. 
    For ease of notation, let $\sigma = \rho {\abs{B}}/{\abs{G}}$.
    Recall that a $\rho$ fraction of recommendations are across groups. For $i \in G \cap S$, \cref{eq:nosever_g} and \cref{eq:noconnect_g} above can be relaxed as
    \begin{align}
        0 &\geq \gamma - \max_{\frac{\mu_G(d_i + 1)}{d_i + 1} \leq v \leq  \frac{\mu_G(d_i - 1)}{d_i - 1} } {g_0 v} \paren{1 - v}^{d_i - k - 1}\paren{1 - \frac{\mu_G(\bar d_g)}{\bar d_G}}^{(1-\sigma) k} \paren{1 - \frac{\mu_B(k)}{k}}^{\sigma k}, \;\; \text{and} \label{eq:nosever_g2}\\
        0 &\geq \frac{g_0 \mu_G(d_i+1)}{d_i + 1} \paren{1 - \frac{\mu_G(d_i+1)}{d_i+1}}^{d_i - k} \paren{1 - \frac{\mu_G(\munderbar d_G)}{\munderbar d_G}}^{(1-\sigma ) k} \paren{1 - \frac{\mu_B(k)}{k}}^{\sigma k} - \gamma.\label{eq:noconnect_g2} 
    \end{align}
    In particular, for $i \in G \cap S$, we can rearrange \cref{eq:noconnect_g2} as
    \begin{align*}
        d_i &\geq \frac{g_0 \mu_G(d_i + 1)}{\gamma} \paren{1 - \frac{\mu_G(d_i+1)}{d_i + 1}}^{d_i - k} \paren{1 - \frac{\mu_G(\munderbar d_G)}{\munderbar d_G}}^{(1-\sigma ) k} \paren{1 - \frac{\mu_B(k)}{k}}^{\sigma k} - 1, 
    \end{align*}
    and since any individual $i \in G \cap S$ with degree $\dglower$ must also satisfy this constraint, we have
    \begin{align}
        \dglower &\geq \frac{g_0 \mu_G(\dglower + 1)}{\gamma} \paren{1 - \frac{\mu_G(\dglower + 1)}{\dglower + 1}}^{\dglower - k} \paren{1 - \frac{\mu_B(k)}{k}}^{\sigma k} \paren{1 - \frac{\mu_G(\munderbar d_G)}{\munderbar d_G}}^{(1-\sigma) k} - 1 \nonumber \\
        &\geq \frac{g_0 \mu_G(\dglower)}{\gamma} \paren{1 - \frac{\mu_G(\dglower)}{\dglower}}^{\dglower - \sigma k} \paren{1 - \frac{\mu_B(k)}{k}}^{\sigma k} - 1\nonumber \\
        &\geq \frac{g_0 \mu_G(\dglower) e^{-\mu_G}}{\gamma} \paren{1 - \frac{\mu_B(k)}{k}}^{\sigma k} -1 \nonumber \\
        &\geq \frac{g_0 \mu_G(\dglower) e^{-\mu_G}}{\gamma} (1 - \varepsilon) -1 \label{eq:dblower} 
    \end{align}
    The second inequality comes from combining the second and fourth terms in the RHS product; the third inequality comes from the fact that $(1-\mu_G(\dglower)/\dglower)^{d_i-\sigma k} \geq e^{-\mu_G}$; the fourth inequality comes from the fact that we can set $b_0$ large enough that $\mu_B(k) \leq \mu_B \leq \varepsilon k$.
    Notice that there exists some $\bar \gamma$ so that for all $\gamma \leq \bar \gamma$, it holds $\dglower \geq 2$, so \cref{eq:nosever_g} holds for all $i \in G \cap S$.
    Moreover, notice that there exists some $\bar \gamma$ so that for all $\gamma < \bar \gamma$ it holds $\mu_G(\dglower) = \mu_G$ so that the expression simplifies to
    \begin{align*}
        \dglower \geq \frac{g_0 \mu_G e^{-\mu_G}}{\gamma} (1 - \varepsilon) -1.
    \end{align*}
    Throughout, we just need to ensure that $\munderbar{\gamma} < \bar \gamma$, which can be achieved by ensuring $b_0$ is sufficiently close to 1 (so that $\munderbar{\gamma} = b_0(1- b_0 - b_1)$ is sufficiently small).
    
    To derive an upper bound for $\dgupper$, we can use the fact that $(1-\mu_B(k)/k) \leq 1$ and $(1-\mu_B(\dgupper)/\dgupper) \leq 1$ along with \cref{eq:nosever_g2} to imply
    \begin{align*}
        \dgupper &\leq \frac{g_0 \mu_G}{\gamma} \paren{1 - \frac{\mu_G}{\dgupper + 1}}^{\dgupper - \sigma k - 1 } + 1, 
    \end{align*}
    where the inequality comes from the fact that the RHS expression is monotonic decreasing in $d_i$.
    Now, for all $\varepsilon > 0$, notice that \cref{eq:dblower} implies that there exists some $\bar \gamma$ such that for all $\gamma < \bar \gamma$, 
    \begin{align*}
        \paren{1 - \frac{\mu_G}{\dgupper + 1}}^{\dgupper - \sigma k - 1} &\leq 
        \paren{1 - \frac{\mu_G}{\dglower + 1}}^{\dglower - \sigma k - 1} \\
        &\leq e^{-\mu_G} (1 + \varepsilon),
    \end{align*}
    since the LHS expression above is monotonic decreasing in $\dgupper$, $\dglower \to \infty$ as $\gamma \to 0$ and, plugging in \cref{eq:dblower},
    \begin{align*}
        \lim_{\dglower \to \infty} \paren{1 - \frac{\mu_G}{\dglower - 1}}^{\dglower - \sigma k - 1} = e^{-\mu_G}.
    \end{align*}
    Plugging this into the previous expression for $\gamma$ in the same range, we have
    \begin{align*}
        \dgupper &\leq \frac{g_0 \mu_G e^{-\mu_G}}{\gamma}(1+\varepsilon) + 1.
    \end{align*}

    Step 2. Define $\munderbar{u}_B, \bar{u}_B$ be lower and upper bounds on the utilities of individuals in $B$ for all $E \in \cE$.
    Define $\munderbar{u}_G, \bar{u}_G$ analogously for $G \cap S$.
    (Note that the definitions are slightly different: the bounds for $B$ holds for \textit{all} members of $B$ and the bounds for $G$ only hold for the intersection with $S$.)
    Using $\munderbar d_G, \bar d_G$, we can apply the utility expressions to show
    \begin{align*}
        \munderbar{u}_G 
        &\geq 1 - g_0 \max_{v \in \{ \dglower, \dots, \dgupper \}} \paren{1 - \frac{\mu_G(v+1)}{v + 1}}^{v - k} \paren{1 - \frac{\mu_B(k)}{k}}^{\sigma k} \paren{1 - \frac{\mu_G(\bar d_G)}{\bar d_G}}^{(1-\sigma ) k}- \gamma (\dgupper - k) \\
        &\geq 1 - g_0 \max_{v \in \{ \dglower, \dots, \dgupper \}} \paren{1 - \frac{\mu_G(v+1)}{v + 1}}^{v - k}  - \gamma (\dgupper - k).
    \end{align*}
    Now, if $k \geq 2$, since $(1-\mu_G(v+1)/(v+1))^{v-k}$ is monotonic decreasing in $v$ and approaching $e^{-\mu_G}$ as $\gamma \to 0$, there exists $\bar \gamma$ such that 
    \begin{align*}
        \max_{v \in \{ \dglower, \dots, \dgupper \}} \paren{1 - \frac{\mu_G(v+1)}{v + 1}}^{v - k} \paren{1 - \frac{\mu_B(k)}{k}}^{\sigma k} \leq e^{-\mu_G} (1 + \varepsilon).
    \end{align*}
    for all $\gamma < \bar \gamma$.
    If $k = 1$, then the expression is monotonic increasing and approaching $e^{-\mu_G} \leq e^{-\mu_G}(1 + \varepsilon)$. Thus, 
    \begin{align*}
        \munderbar{u}_G & \geq 1 - g_0 e^{-\mu_G} (1 + \varepsilon) - \gamma(\dgupper - k) \\
        &\geq 1 - g_0 e^{-\mu_G} (1 + \varepsilon) - g_0 \mu_G e^{-\mu_G}(1 + \varepsilon) \\
        &\geq 1 - g_0 (1 + \mu_G) e^{-\mu_G} - \varepsilon
    \end{align*}
    where the $\varepsilon$ on each line may be different.
    For $\bar{u}_G$
    \begin{align*}
        \bar{u}_G 
        &\leq 1 - g_0 \min_{v \in \{ \dglower, \dots, \dgupper \}} \paren{1 - \frac{\mu_G(v-1)}{ v - 1 }}^{v - k} \paren{1 - \frac{\mu_B(k)}{k}}^{\sigma k} \paren{1 - \frac{\mu_G(\dglower) }{\munderbar d_G}}^{(1-\sigma ) k} - \gamma (\dglower - k) \\
        &\leq 1 - g_0 \min_{v \in \{ \dglower, \dots, \dgupper \}} \paren{1 - \frac{\mu_G(v-1)}{ v - 1 }}^{v - k} \paren{1 - \frac{\mu_B(k)}{k}}^{\sigma k} \paren{1 - \frac{\mu_G(v-1)}{v- 1}}^{(1-\sigma) k} - \gamma (\dglower - k) \\
        &\leq 1 - g_0 \min_{v \in \{ \dglower, \dots, \dgupper \}} \paren{1 - \frac{\mu_G(v-1)}{ v - 1 }}^{v - \sigma k} \paren{1 - \frac{\mu_B(k)}{k}}^{\sigma k}  - \gamma (\dglower - k) \\
        &= 1 - g_0 e^{-\mu_G}(1 - \varepsilon) \paren{1 - \frac{\mu_B(k)}{k}}^{\sigma k}  - \gamma (\dglower - k) \\
        &\leq 1 - g_0 e^{-\mu_G}(1 - \varepsilon)  - \gamma (\dglower - k) \\
        &\leq 1 -  g_0 e^{-\mu_G} (1 - \varepsilon)  - \gamma (k + 1) - g_0 \mu_G e^{-\mu_G}(1 - \varepsilon) \\
        &\leq 1 -  g_0(1 + \mu_G) e^{-\mu_G} - \gamma (k + 1) + \varepsilon  
    \end{align*}
    where again the $\varepsilon$ on each line may be different.
    %
    For the blue group,
    \begin{align*}
        \munderbar{u}_B &\geq  1 - b_0  \paren{1 - \frac{\mu_G(\dgupper)}{\bar d_G}}^{\rho k} \paren{1 - \frac{\mu_B(k)}{k}}^{(1-\rho) k}, \text{ and} \\
        \bar{u}_B &\leq 1 - b_0 \paren{1 - \frac{\mu_G(\dglower)}{\munderbar d_G}}^{\rho k} \paren{1 - \frac{\mu_B(k)}{k}}^{(1-\rho) k}
    \end{align*}
    which implies
    \begin{align*}
        \munderbar{u}_B 
        &\geq 1 - b_0 \paren{1 - \frac{\gamma }{g_0 e^{-\mu_G} (1 + \varepsilon) + \gamma \mu_G^{-1}}}^{\rho k} \\
        &\geq 1 - b_0 \paren{1 - \frac{\gamma }{g_0 e^{-\mu_G} + \gamma \mu_G^{-1}}}^{\rho k} - \varepsilon
    \end{align*}
    and
    \begin{align*}
        \bar{u}_B 
        &\leq 1 - b_0 \paren{1 - \frac{\gamma}{ g_0 e^{-\mu_G}  (1 - \varepsilon) - \gamma \mu_G^{-1}}}^{\rho k}(1 - \varepsilon) \\
        &\leq 1 - b_0 \paren{1 - \frac{\gamma}{ g_0 e^{-\mu_G}  - \gamma \mu_G^{-1}}}^{\rho k} + \varepsilon.
    \end{align*}
    Finally, we set $n$ large enough that the effect of $[n] \setminus S$ on average utility is sufficiently small that the utility bounds hold with $\varepsilon$ set appropriately above.
    \end{proof}
    \begin{proof}[Proof of \Cref{prop:cg-vs-ni}]
    First, we prove \cref{eq:sw_k_utilitarian}. We apply \Cref{lem:utilbounds} and compare a lower bound on social welfare when $k > 0$ to the upper bound for no recommendations and show that the former is greater. 
    Holding $n, b,g,\gamma$ fixed, we will denote the associated utility bounds as $\ubupper^{(k)}, \ubupper^{(0)}$ respectively. 
    To prove the result, it is sufficient to show
    \begin{align}
        \abs{B}({\ublower^{(k)} - \ubupper^{(0)} }) + \abs{G}({\uglower^{(k)} - \ugupper^{(0)}}) > 0, \label{eq:positivesw}
    \end{align}
    which states that a lower bound on social welfare under $k$ recommendations is greater than an upper bound on social welfare without recommendations.
    By plugging in the bounds from \Cref{lem:utilbounds}, for all $\varepsilon > 0$, sufficiently large $\munderbar{b}_0, \munderbar{n}$,
    \begin{align*}
        &\ublower^{(k)} - \ubupper^{(0)} \geq  b_0 \paren{1 - \paren{1 - \frac{\gamma }{g_0 e^{-\mu_G} + \gamma \mu_G^{-1}}}^{k} } - \varepsilon, \text{ and} \\
        &\uglower^{(k)} - \ugupper^{(0)} \geq  -\gamma(k+1)  - \varepsilon.
    \end{align*}
    Putting these together, \cref{eq:positivesw} holds if 
    \begin{equation}
    \label{eq:goal-5-1}
        b_0 > \frac{\abs{G}}{\abs{B}} \frac{\gamma (k+1)}{1 - \paren{1 - \frac{\gamma}{g_0 e^{-\mu_G} + \gamma \mu_G^{-1}}}^{k}},
    \end{equation}
    where we plug in $\ublower^{(k)} - \ubupper^{(0)}$ and $\uglower^{(k)} - \ugupper^{(0)}$ into \cref{eq:positivesw} and solve for $b_0$.
    %
    Note that
    \[
    \paren{1 - \frac{\gamma}{g_0 e^{-\mu_G} + \gamma \mu_G^{-1}}}^{k} \le 1 - \frac{\gamma}{g_0 e^{-\mu_G} + \gamma \mu_G^{-1}}.
    \]
    Therefore,
   \begin{align}
    \frac{\abs{G}}{\abs{B}} \frac{\gamma (k+1)}{1 - \paren{1 - \frac{\gamma}{g_0 e^{-\mu_G} + \gamma \mu_G^{-1}}}^{k}} \le \frac{\abs{G}}{\abs{B}} (k+1)(g_0 e^{-\mu_G} + \gamma \mu_G^{-1}). \label{eq:b0_lower}
   \end{align} 
    Thus, a sufficient condition for \eqref{eq:goal-5-1} to hold is $$b_0 > \frac{\abs{G}}{\abs{B}} (k+1)(g_0 e^{-\mu_G} + \gamma \mu_G^{-1}).$$
    Finally, setting $\bar g_0$ so that $g_0e^{-\mu_G} \leq g_0 e^{-(1-g_0)C} < \abs{B}/(\abs{G}4(k+1))$ (for $C$ the maximum number of possible exogenous opportunities) and $\gamma< \abs{B} \mu_G /\abs{G}(4(k+1)) \leq \bar \gamma$ proves that the RHS of \cref{eq:b0_lower} is strictly less than 1, so setting $\munderbar{b}_0$ equal to the maximum of the RHS and the $\munderbar{b}_0$ necessary so that $\munderbar{\gamma} < \bar \gamma$ in \Cref{lem:utilbounds} is sufficient to prove  \cref{eq:sw_k_utilitarian}.

    Next, we prove \cref{eq:sw_k_rawls}. First note that by \Cref{lem:no-blue-edges}, under $E_0$, members of $B$ form no connections. Thus, since $b_0 > g_0$, 
    \begin{align*}
        \argmin_{i \in [n]} u_i(E_0) \subseteq B,
    \end{align*}
    and 
    \begin{align*}
        \min_{i \in [n]} u_i(E_0) = 1 - b_0.
    \end{align*}
    On the other hand, under recommendations (and the parameters in the statement of the result), every individual $i \in [n]$ forms at least $k$ connections, and therefore must have utility no less than 
    \begin{align*}
        1 - p_{i0} \prod_{j\; : \; (i,j) \in Q} \paren{1 - \frac{\mu_j(d_j)}{d_j}} > 1 - p_{i0} \geq 1 - b_0.
    \end{align*}
    \end{proof}

    \begin{proof}[Proof of \Cref{prop:crossgroupgood}]
    We apply \Cref{lem:utilbounds} and compare the lower bound on social welfare with equilibria induced under $\rho$ versus $\rho'$ and show the former is greater.
    Holding $n, b,g,\gamma$ and a sufficiently small $\varepsilon > 0$ fixed, we will denote the associated utility bounds for sets $S^{(\rho)}, S^{(\rho')}$ satisfying \Cref{lem:reciprocity-w-recs} for $\varepsilon$ as $\ubupper^{(\rho)}, \ubupper^{(\rho')}$ respectively. 
    %
    Notice that, for large enough $n$, this is implied by
    \begin{equation}
    \label{eq:goal-5-2}
        \abs{B}({\ublower^{(\rho)} - \ubupper^{(\rho')} }) + \abs{G}({\uglower^{(\rho)} - \ugupper^{(\rho')}}) > 0.
    \end{equation}
    Applying the bounds from \Cref{lem:utilbounds}, for all $\varepsilon > 0$
    \begin{align*}
        &\ublower^{(\rho)} - \ubupper^{(\rho')} 
        \geq b_0 \paren{\paren{1 - \frac{\gamma}{ g_0 e^{-\mu_G}  - \gamma \mu_G^{-1}}}^{\rho' k} - \paren{1 - \frac{\gamma }{g_0 e^{-\mu_G} + \gamma \mu_G^{-1}}}^{\rho k}} - \varepsilon, \text{ and} \\
        &\uglower^{(\rho)} - \ugupper^{(\rho')} \geq -\gamma(k+1) - \varepsilon.
    \end{align*}
    Putting these together, we have for all $\varepsilon > 0$,
    \begin{align*}
        &\abs{B} (\ublower^{(\rho)} - \ubupper^{(\rho')}) + \abs{G} (\uglower^{(\rho)} - \ugupper^{(\rho')}) \\
        &\geq \abs{B}b_0 \paren{\paren{1 - \frac{\gamma}{ g_0 e^{-\mu_G}  - \gamma \mu_G^{-1}}}^{\rho' k} - \paren{1 - \frac{\gamma }{g_0 e^{-\mu_G} + \gamma \mu_G^{-1}}}^{\rho k}} - \abs{G}\gamma (k+1) - \varepsilon. 
    \end{align*}
    Thus, to guarantee \eqref{eq:goal-5-2}, it suffices to make sure
    \[
    \abs{B}b_0 \paren{\paren{1 - \frac{\gamma}{ g_0 e^{-\mu_G}  - \gamma \mu_G^{-1}}}^{\rho' k} - \paren{1 - \frac{\gamma }{g_0 e^{-\mu_G} + \gamma \mu_G^{-1}}}^{\rho k}} - \abs{G}\gamma (k+1) > 0,
    \]
    which is equivalent to
    \begin{align}
        \rho'  < \frac{\log \paren{ \paren{1 - \frac{\gamma }{g_0 e^{-\mu_G} + \gamma \mu_G^{-1}}}^{\rho k}  + \abs{G}\abs{B}^{-1}b_0^{-1}\gamma(k+1)}}{k \log \paren{1 - \frac{\gamma}{ g_0 e^{-\mu_G}  - \gamma \mu_G^{-1}}}}. \label{eq:rho}
    \end{align}
    We need the RHS to be positive for there to exist some $\rho'$ satisfying the statement of the result.
    Notice that 
    \begin{align*}
        {1 - \frac{1}{ g_0 e^{-\mu_G}\gamma^{-1}  - \mu_G^{-1}}} < 1 \iff \gamma < g_0 \mu_G e^{-\mu_G}
    \end{align*}
    which implies the logarithm is negative.
    Thus, the RHS expression in \cref{eq:rho} is greater than zero if     
    \begin{align*}
         \paren{1 - \frac{\gamma }{g_0 e^{-\mu_G} + \gamma \mu_G^{-1}}}^{\rho k}  + \frac{\abs{G}}{\abs{B}}\frac{\gamma(k+1)}{b_0} < 1
    \end{align*}
    which holds for all $\rho k \geq 1$ if
    \begin{align*}
        b_0 > \frac{\abs{G}}{\abs{B}}\paren{\frac{g_0}{ e^{\mu_G}} + \frac{\gamma}{\mu_G}}(k+1).
    \end{align*}
    This RHS expression is less than 1 for all $g_0$ such that 
    \begin{align*}
        g_0 < \frac{\abs{G}(k+1) - \abs{B}}{\abs{G}(k+1) + \abs{B}}.
    \end{align*}

    Our proof of the second statement (about the minimum utility) relies on arguing that, first, the minimum is achieved by blue group members and, second, that the blue group member utilities are greater, for all equilibria under $\rho k$ cross-group connections than under $\rho' k$ cross-group connections.
    Notice, for all $i \in B$,
    \begin{align*}
        u_i(E_\rho) \leq 1 - b_0 \paren{1 - \frac{\mu_G}{k}}^{\rho k} + \varepsilon.
    \end{align*}
    This is a larger bound than in \Cref{lem:utilbounds} because we are reasoning about all $i \in B$, not just $i \in B \cap S$.
    On the other hand, for all $j \in G$,
    \begin{align*}
        u_j(E_\rho) \geq 1 - g_0. 
    \end{align*}
    Thus, for $u_i(E_\rho) < u_j(E_\rho)$ if
    \begin{align*}
        b_0 > \frac{g_0}{\paren{1 - \frac{\mu_G}{k}}^{\rho k}}.
    \end{align*}
    Since the denominator in the RHS expression is monotone decreasing in $k$, it is lower-bounded by its limit as $k \to \infty$, which is $e^{-\rho \mu_G}$. Thus, a sufficient condition for the RHS being less than 1 is $g_0 e^{\mu_G} < 1$.

    Next, we will argue that $\munderbar{u}_B^{(\rho)} - \bar{u}_B^{(\rho')}\geq 0$. By the same reasoning as above, this holds if 
    \begin{align*}
        \rho'  < \frac{\log \paren{ \paren{1 - \frac{\gamma }{g_0 e^{-\mu_G} + \gamma \mu_G^{-1}}}^{\rho k} }}{k \log \paren{1 - \frac{\gamma}{ g_0 e^{-\mu_G}  - \gamma \mu_G^{-1}}}}.
    \end{align*}
    Finally, the RHS is positive if 
    \begin{align*}
        {\frac{\gamma }{g_0 e^{-\mu_G} + \gamma \mu_G^{-1}}} > 0
    \end{align*}
    which is true for all positive $\gamma, g_0, \mu_G$.
    
    \end{proof}





\end{document}